\documentclass[11pt]{myclass}

\usepackage{amssymb,amsmath,amsfonts,pdfpages,color, url}
\usepackage{booktabs}
\usepackage{wrapfig}
\usepackage{amsmath}
\usepackage{multirow}
\newtheorem{cor}[theorem]{Corollary}

\newtheorem{definition}[theorem]{Definition}
\newtheorem{exa}{Example}[section]

\newcommand{\R}{\ensuremath{\mathbb{R}}}

\def\argmin{\mathop{\rm argmin}}
\def\SMA{\mathop{\rm SMA}}
\def\slfs{\mathop{\rm slfs}}
\def\MA{\mathop{\rm MA}}
\def\lfs{\mathop{\rm lfs}}
\def\argmax{\mathop{\rm argmax}}
\def\reg{\mathop{\rm reg}}

\def\sign{\mathop{\rm sign}}
\def\diam{\mathop{\rm diam}}
\def\Gsim{\mathop{\Gamma}}

\newcommand{\eps}{\varepsilon}

\newcommand{\dQ}{\ensuremath{\mathtt{d}_{Q}}}
\newcommand{\dF}{\ensuremath{\mathtt{d}_{F}}}
\newcommand{\dH}{\ensuremath{\mathtt{d}_{H}}}
\newcommand{\dQmD}{\ensuremath{\mathtt{d}^{\textsf{mD}}_{Q}}}

\newcommand{\Frechet}{Fr\'echet\xspace}

\begin{document}

\title{Orientation-Preserving Vectorized Distance Between Curves}

\author{Jeff M. Phillips and Hasan Pourmahmood-Aghababa \\ University of Utah \\ \texttt{jeffp@cs.utah.edu} and \texttt{h.pourmahmoodaghababa@utah.edu}}

\maketitle

\begin{abstract}
We introduce an orientation-preserving landmark-based distance for continuous curves, which can be viewed as an alternative to the \Frechet or Dynamic Time Warping distances.  This measure retains many of the properties of those measures, and we prove some relations, but can be interpreted as a Euclidean distance in a particular vector space.   Hence it is significantly easier to use, faster for general nearest neighbor queries, and allows easier access to classification results than those measures.  It is based on the \emph{signed} distance function to the curves or other objects from a fixed set of landmark points.  
We also prove new stability properties with respect to the choice of landmark points, and along the way introduce a concept called signed local feature size (slfs) which parameterizes these notions.  Slfs explains the complexity of shapes such as non-closed curves where the notion of local orientation is in dispute -- but is more general than the well-known concept of (unsigned) local feature size, and is for instance infinite for closed simple curves.  
Altogether, this work provides a novel, simple, and powerful method for oriented shape similarity and analysis.   
\end{abstract}

\keywordlist{Computational geometry, learning on structured data, feature mapping, medial axis, feature size.} \\
\textbf{Venue:} To be appeared in Mathematical and Scientific Machine Learning (MSML), August 2021. 

\let\thefootnote\relax\footnotetext{Jeff Phillips thanks his support from NSF CCF-1350888, CNS- 1514520, CNS-1564287, and IIS-1816149. The authors thank Micka\"el Buchet for helping uncover a technical issue in a previous version.}

\newpage
\setcounter{page}{1}

\section{Introduction}

The \Frechet distance~\cite{AG95} is a very popular distance between curves; it has spurred significantly practical work improving its empirical computational time~\cite{bringmann2019walking} (including a recent GIS Cup challenge~\cite{werner2018acm,baldus2017fast,buchin2017efficient,dutsch2017filter} and inclusion in sklearn) and has been the subject of much algorithmic studies on its computational complexity~\cite{Bri14,ABKS14,buchin2017four}.  While in some practical settings it can be computed in near-linear time~\cite{driemel2012approximating}, there exists settings where it may require near-quadratic time -- essentially reverting to dynamic programming~\cite{Bri14}.  

The interest in studying the \Frechet distance (and similar distances like the discrete \Frechet distance~\cite{TEHM1994}, Dynamic Time Warping~\cite{lemire2009faster}, edit distance with real penalties~\cite{edr}) has grown recently due to the very real desire to apply them to data analysis.  Large corpuses of trajectories have arisen through collection of GPS traces of people~\cite{geolife-gps-trajectory-dataset-user-guide}, vehicles~\cite{GTDS2016}, or animals~\cite{buchin2019klcluster}, as well as other shapes such as letters~\cite{williams2007primitive}, time series~\cite{DKS16}, and more general shapes~\cite{AKW04}.  What is common about these measures, and what separates them from alternatives such as the Hausdorff distance is that they capture the direction or orientation of the object.  However, this enforcing of an ordering seems to be directly tied to the near-quadratic hardness results~\cite{Bri14}, deeply linked with other tasks like edit distance~\cite{BK15,backurs2015edit}.

Moreover, for data analysis on large data sets, not only is fast computation needed, but so are other operations like fast nearest-neighbor search or inner products.  While a lot of progress has been made in the case of \Frechet distance and the like~\cite{GMMC19, SLB2018, XLP2017,BCG11, DS17, DPS19, FFK20}, these operations are still comparatively slow and limited.  For instance, some of the best fast nearest neighbor search for LSH for discrete \Frechet distance on $n$ curves with $m$ waypoints can answer a query within $1+\eps$ distance using $O(m)$ time, but requiring $n \cdot O((1/\eps)^m)$  space~\cite{FFK20}; or if we reduce the space to something reasonable like $O(n \log n +mn)$, then a query in $O(m \log n)$ time can provide only an $O(m)$ approximation~\cite{DS17}.  

On the other hand, fast nearest neighbor search for Euclidean distance is far more mature, with better LSH bounds, but also quite practical algorithms~\cite{KGraph,FALCONN}.  Moreover, most machine learning libraries assume as input Euclidean data, or for other standard data types like images~\cite{imagenet} or text~\cite{Art3,Mik1} have sophisticated methods to map to Euclidean space.  
However, \Frechet distance is known not to be embeddable into a Euclidean vector space without quite high distortion~\cite{indyk1998approximate,DS17}.  

\paragraph{Embeddings first.} This paper on the other hand starts with the goal of embedding ordered/oriented curve (and shape) data to a Euclidean vector space, where inner products are natural, fast nearest neighbor search is easily available, and it can directly be dropped into any machine learning or other data analysis frameworks. 

This builds on recent work with a similar goals for halfspaces, curves, and other shapes~\cite{PT19a,PT20}.  But that work did not encode orientation.  
This orientation-preserving aspect of these distances is clearly important for some applications; it is needed to say distinguish someone going to work versus returning from work.  

Why might \Frechet be better for data analysis than discrete \Frechet or DTW or the many other distances?  One can potentially point to long segments and no need to discretize, or (quasi-)metric properties.  Regardless, an equalizer is in determining how well a distance models data is the prediction error for classification tasks; such tasks demonstrate how well the distances encode what truely matters on real tasks.  The previous vectorized representations matched or outperformed a dozen other measures~\cite{PT19a}.  In this paper, we show an oriented distance performs similarly on general tasks, but when orientation is essential, does significantly better than non-orientation preserving measures.  
Moreover, by extending properties from similar, but non-orientable vectorized distances~\cite{PT19a,PT20}, our proposed distance inherits metric properties, can handle long segments, and \emph{also} captures curve orientation.   

More specifically, our approach assumes all objects are in a bounded domain $\Omega$ a subset of $\mathbb{R}^d$ (typically $\mathbb{R}^2$).  This domain contains a set of landmark points $Q$, which might constitute a continuous uniform measure over $\Omega$, or a finite sample approximation of that distribution.  With respect to an object $\gamma$, each landmark $q_i \in Q$ generates a value $v_{q_i}(\gamma)$.  Each of these values $v_{q_i}(\gamma)$ can correspond with the $i$th coordinate in a vector $v_Q(\gamma)$, which is finite (with $|Q|=n$-dimensions) if $Q$ is finite.  Then the distance between two objects $\gamma$ and $\gamma'$ is the square-root of the average squared distance of these values -- or the Euclidean distance of the finite vectors
\[
\dQ(\gamma, \gamma') = \|v_Q(\gamma) - v_Q(\gamma')\|.  
\]
The innovation of this paper is in the definition of the value $v_{q_i}(\cdot)$ and the implications around that choice. In particular in previous works by \cite{PT19a,PT20}, in this framework, this had been (mostly) set as the \emph{unsigned} minDist function: $v^{\mathsf{mD}}_q(\gamma) = \min_{p \in \gamma}\|p - q\|$.  In this paper we alter this definition to not only capture the distance to the shape $\gamma$, but in allowing negative values to also capture the orientation of it.  

This new definition leads to many interesting structural properties about shapes.  These include:
\begin{itemize}
\item
We show that $\dF$ is stable up to \Frechet perturbations of the curves.  In particular, for two curves $\gamma, \gamma'$ which are closed, then 
$\dQ(\gamma,\gamma') < \dF(\gamma, \gamma')$, 
where $\dF$ is the \Frechet distance.  
Moreover, for an $l^\infty$ variant of the distance $\dQ^\infty$ (see Definition \ref{d4}), 
and the curves are closed, we obtain  $\dF(\gamma,\gamma') = \dQ^{\infty}(\gamma,\gamma')$.  
Thus $\dQ$ captures orientation.  
In contrast for a class of curves we show ${\dQmD}^{,\infty}$ can equal Hausdorff distance, so this older version explicitly does not capture orientation. 

\item 
We introduce a new stability notion called the \emph{signed local feature size} for oriented curves and other shapes $\gamma$.  While it does not rely on the new signed $v_{q_i}(\gamma) : \mathbb{R}^d \to \mathbb{R}$ distance function, it captures a scale under which $v_{q_i}(\gamma)$ is stable. 
Unlike its unsigned counterpart (\emph{local feature size}, which plays a prominent role in shape reconstruction~\cite{ACK01,cheng2012delaunay,CC12} and computational topology~\cite{CFLMRW14,ChazalCohen-SteinerMerigot2011,CL05}), the signed local feature size is infinite for closed simple curves. This in turn implies $\dQ$ is stable with respect to $Q$ for all closed simple curves, where as notions which depend on (unsigned) local feature size (e.g., medial axis) are not. 

\item 
For curves with boundary, the signed local feature size requires a more intricate treatment. 
We show that when the signed local feature size $\delta$ is positive but finite ($0 < \delta < \infty$) then we can set a scale parameter $\sigma$ in the definition of $v_q$ (denoted $v_q^\sigma$) and in $\dQ$ (denoted $\dQ^\sigma$) so when $\sigma < \delta / (4(1+\sqrt{\ln(2/\eps)})$, then the signed distance function $v_i^\sigma$ is stable up to value $\eps$.  
\end{itemize}

Altogether, these results build and analyze a new vectorized, and sketchable distance between curves (or other geometric objects) which captures orientation like \Frechet (or dynamic time warping, and other popular measures), but avoids all of the complications when actually used.  As we demonstrate, fast nearest neighbor search, machine learning, clustering, etc are all now very easy.


By a {\it curve} we mean the image of a continuous non-constant mapping $\gamma: [0,1] \to \mathbb{R}^2$; we simply use $\gamma$ to refer to these curves. Two curves are equivalent if they can be reparameterized by an increasing monotone function to be identical. Hence there are two equivalence classes corresponding to the same trace of a curve, these two classes correspond to the direction of the curve.
A curve $\gamma$ is {\it closed} if $\gamma(0)= \gamma(1)$.  
It is {\it simple} if the mapping does not cross itself, i.e. $\gamma$ is one-to-one function. 

Let $\Gamma$ be the class of all simple curves $\gamma$ in $\mathbb{R}^2$ with the property that at almost every point $p$ on $\gamma$, considering the direction of the curve, there is a unique normal vector $n_p$ at $p$, i.e. there is a tangent line almost everywhere on $\gamma$. Such points are called {\it regular points} of $\gamma$ and the set of regular points of $\gamma$ is denoted by $\reg(\gamma)$. Points of $\gamma \setminus \reg(\gamma)$ are called {\it critical points} of $\gamma$.  The terminology ``almost every point" means that the Lebesgue measure of those $t\in [0,1]$ such that $\gamma(t)$ is a critical point is zero. We also assume that at critical points, which are not endpoints of a non-closed curve, $\gamma$ is left and right differentiable but left and right derivatives are possibly not identical, that is, left and right tangent lines exist. Finally, we assume that non-closed curves in $\Gamma$ have left/right tangent line at endpoints.  These assumptions will guarantee the existence of a unique normal vector at critical points that are defined in Section \ref{S1}.

\paragraph{Baseline distances.}
Important baseline distances are the Hausdorff and \Frechet distances. 
Given two compact sets $A,B \subset \R^d$, the \emph{directed Hausdorff distance} is $\overrightarrow{\dH}(A,B) = \max_{a \in A} \min_{b \in B} \|a-b\|$.  Then the \emph{Hausdorff distance} is defined $\dH(A,B) = \max\{ \overrightarrow{\dH}(A,B), \overrightarrow{\dH}(B,A)\}$.  

The \Frechet distance is defined for curves $\gamma,\gamma'$ with images in $\R^d$.   
Let $\Pi$ be the set of all monotone reparamatrizations (a non-decreasing function $\alpha$ from $[0,1] \to [0,1]$). 
It will be essential to interpret the inverse of $\alpha$ as interpolating continuity; that is, if a value $t$ is a point of discontinuity for $\alpha$ from $a$ to $b$, then the inverse $\alpha^{-1}$ should be $\alpha^{-1}(t') = t$ for all $t' \in [a,b]$.  Together, this allows $\alpha$ (and $\alpha^{-1}$) to represent a continuous curve in $[0,1] \times [0,1]$ that starts at $(0,0)$ and ends at $(1,1)$ while never decreasing either coordinate; importantly, it can move vertically or horizontally.  
Then the \emph{\Frechet distance} is 
\[
\dF(\gamma,\gamma') = \inf_{\alpha \in \Pi} \max\{\sup_{t \in [0,1)} \|\gamma(t) - \gamma'(\alpha(t))\|, \sup_{t \in [0,1)} \|\gamma(\alpha^{-1}(t)) - \gamma'(t)\|\}.  
\]
We can similarly define the \Frechet distance for closed oriented curves (see also~\cite{AKW04,SVY14}); it is useful to interpret this parameterization of curve $\gamma$ as measuring \emph{arclength}.
Given an arbitrary point $c_0 \in \gamma$, then $\gamma(t)$ for $t \in [0,1)$ indicates the distance along the curve from $c_0$ in a specified direction, divided by the total arclength.  Let $\Pi^\circ$ denote the set of all monotone, cyclic parameterizations; now $\alpha \in \Pi^\circ$ is a function from $[0,1) \to [0,1)$ where it is non-decreasing everywhere except for exactly one value $a$ where $\alpha(a) = 0$ and $\lim_{t \nearrow a} \alpha(t) = 1$.  Again, $\alpha^{-1} \in \Pi^{\circ}$ has the same form, and interpolates the discontinuities with segments of the constant function.  Then the \Frechet distance for oriented closed curves is defined $\dF(\gamma,\gamma') = \inf_{\alpha \in \Pi^\circ} \max\{ \sup_{t \in [0,1)} \|\gamma(t) - \gamma'(\alpha(t))\|, \, \sup_{t \in [0,1)} \|\gamma(\alpha^{-1}(t)) - \gamma'(t) \|\}$.  
Oriented closed curves are important for modeling boundary of shapes and levelsets~\cite{LC87}, orientation determines inside from outside.

\subsection{New Definitions for Orientation-Preserving Distance} \label{S1}  \label{sec:def}

We introduce a feature mapping based on some landmark set $Q$ as one of the core definitions of this paper. We will employ the notation $\langle \cdot, \cdot \rangle$ for the usual inner product in a Euclidean space.

\begin{definition}[Feature Mapping] \label{d6} 
Let $\gamma \in \Gamma$, $Q$ be a finite subset of $\mathbb{R}^2$ and $\sigma >0$. For each $q \in Q$ let $p = \argmin_{p' \in \gamma} \|q-p'\|$. If $p$ is not an endpoint of $\gamma$, we define 
\[
v_q^{\sigma}(\gamma)=  \frac{1}{\sigma} \langle n_{p}(q), q-p\rangle e^{- \frac{\|q-p\|^2}{\sigma^2}}.
\]
Otherwise (for endpoints) we set  
\[
v_q^{\sigma}(\gamma) =  \frac{1}{\sigma} \langle n_p, \frac{q-p}{\|q-p\|} \rangle \|q\|_{\infty,p} \, e^{- \frac{\|q-p\|^2}{\sigma^2}},
\]
where $\|q\|_{\infty,p}$ is the $l^{\infty}$-norm of $q$ in the coordinate system with axis parallel to $n_p$ and $L$ (tangent line at $p$) and origin at $p$; see Figure \ref{end-l-infty} (Left) for an illustration. 
Figure \ref{Figslfs} shows an example of $v_q^{\sigma}$ over $\R^2$.  
Notice that $\|q\|_{2,p} = \|q-p\|$ and so $\frac{1}{\sqrt{2}} \leq \frac{\|q\|_{\infty,p}}{\|q-p\|} \leq 1$. If $Q=\{q_1, q_2, \ldots, q_n\}$, setting $v_i^{\sigma}(\gamma) = v_{q_i}^{\sigma}(\gamma)$ we obtain a {\it feature mapping} $v_Q^{\sigma}: \Gsim \to \mathbb{R}^n$ defined by $v_Q^{\sigma}(\gamma)= (v_1^{\sigma}(\gamma), \cdots, v_n^{\sigma}(\gamma))$. (We will drop the superscript $\sigma$ afterwards, unless otherwise specified.)
\end{definition} 

\begin{figure}[h] 
\includegraphics[width=1 \textwidth]{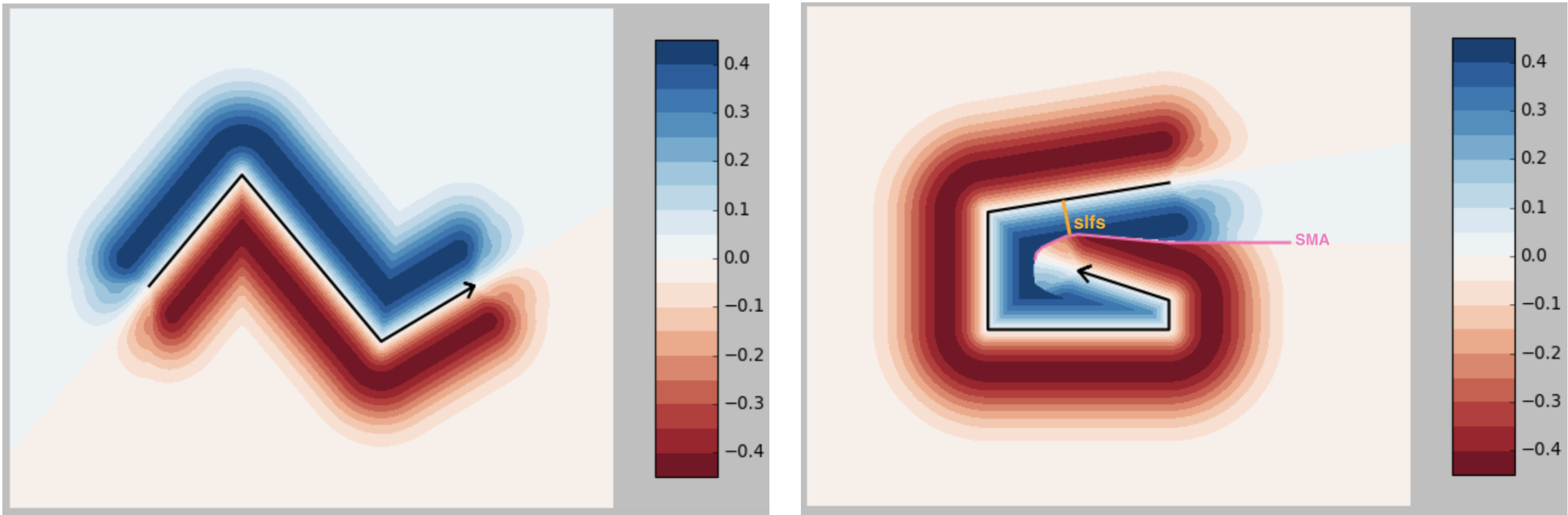}
\caption{\label{Figslfs}Left:  Example of feature mapping $v_q$ for curve with sidedness encoded by positive/negative values. Right: Signed medial axis (SMA - in pink) as the only place with large discontinuity in $v_q$ and signed local feature size (slfs - in orange).}
\end{figure}

The inner product in $v_q^\sigma(\gamma)$ captures the sidedness of $q$ with respect to its nearest point $p \in \gamma$.  The Gaussian weight dampens function value as a point $q$ becomes further from $\gamma$; if we think of it as representing the magnitude of the sidedness, then as $q$ becomes further from $\gamma$, we want to have less confidence in this value.  
In particular, as discussed below, as $q$ approaches the $\SMA$ (see Definition \ref{d5} below), the bandwidth parameter $\sigma$ can be tuned so this magnitude goes close to $0$, and the discontinuity on the $\SMA$ is bounded.  
Figure \ref{Figslfs} provides examples of feature mapping for two curves where in the second complex picture we have identified SMA. 

\begin{figure}[h] 
\includegraphics[width=0.95 \textwidth]{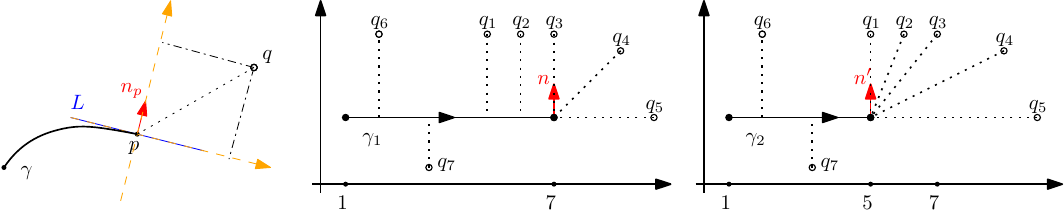}
 
\caption{\label{end-l-infty}$l^{\infty}$-norm at endpoints distinguishes these two simple curves}
\end{figure}

Regarding endpoints, we keep $n_p(q)$ fixed, independent of the choice of $q$.  This ensures that for points $q$ that would be on the ray extending from the endpoint in the same direction, that the function value is $0$, and thus $v_q^\sigma(\gamma)$ does not have a discontinuity here as the sign changes.  
We use $l^{\infty}$-norm (the $\|q\|_{\infty,p}$ term) at these endpoints to obtain the definiteness property of the distance $\dQ^{\sigma}$ in Definition \ref{d4}. Indeed, it enables us to distinguish two curves which are almost identical, such as the two line segments $\gamma_1$ and $\gamma_2$ in Figure \ref{end-l-infty}. If we had alternatively, employed an $l^2$-norm at endpoints (using $\|q\|_{2,p}$, as would be equivalent to the definition of $v_q^\sigma(\gamma)$ at non-endpoints $p$), both curves $\gamma$ and $\gamma'$ would be mapped to the same $v_Q^\sigma$ vectors, i.e. $v_Q^\sigma(\gamma) = v_Q^\sigma(\gamma')$. In contrast, the $l^{\infty}$-norm at endpoints will provide different $v_Q^\sigma$ vectors for them. We will ultimately require that $Q$ is sufficiently dense in order to have definiteness property, discussed next.

\begin{definition}[Orientation Preserving Distance] \label{d4} 
Let $\gamma_1, \gamma_2 \in \Gsim$, $Q=\{q_1, q_2, \ldots, q_n\}$ be a point set in ${\mathbb{R}^2}$, $\sigma>0$ be a positive constant and $p\in [1, \infty]$. The orientation preserving distance of $\gamma_1$ and $\gamma_2$, associated with $Q$, $\sigma$ and $p$, denoted $\dQ^{\sigma,p}(\gamma_1, \gamma_2)$, is the normalized $l^p$-Euclidean distance of two $n$-dimensional feature vectors $v_Q(\gamma_1)$ and $v_Q(\gamma_2)$ in $\mathbb{R}^n$, i.e. for $p \in [1,\infty)$, 
\[
\dQ^{\sigma,p}(\gamma_1, \gamma_2)= \frac{1}{\sqrt[p]{n}} \|v_Q(\gamma_1)- v_Q(\gamma_2)\|_p = \bigg(\frac{1}{n}\sum_{i=1}^n |v_i(\gamma_1)- v_i(\gamma_2)|^p \bigg)^{1/p},
\]
and for $p = \infty$,
\[
\dQ^{\sigma,\infty}(\gamma_1, \gamma_2)= \|v_Q(\gamma_1)- v_Q(\gamma_2)\|_{\infty}= \max_{1 \leq i \leq n} |v_i(\gamma_1)- v_i(\gamma_2)|.
\]
As default we use $\dQ^{\sigma}$ instead of $\dQ^{\sigma,2}$.  
Landmarks $Q$ can be described by a probability distribution $\mu : \R^2 \to \R$, then $v_Q$ is infinite-dimensional, and for $p \in [1, \infty)$ we can define $\dQ^{\sigma,p}(\gamma_1,\gamma_2) = (\int_{q \in \R^2} |v_q(\gamma_1)-v_q(\gamma_2)|^p \mu(q))^{1/p}$.
\end{definition}

Since curves are embedded into a Euclidean space, and the usual $l^p$-norm induces a distance between two curves, the function $\dQ^{\sigma,p}$ enjoys all properties of a metric but the definiteness property. That is, it satisfies triangle inequality, is symmetric, and $\dQ^{\sigma,p}(\gamma_1, \gamma_2)=0$ provided $\gamma_1 = \gamma_2$ (i.e. $\gamma_1$ and $\gamma_2$ have same range and direction). 
However, $\dQ^{\sigma,p}(\gamma_1, \gamma_2)=0$ does not necessarily imply $\gamma_1 = \gamma_2$: consider two curves which overlap, and all landmarks have closest points on the overlap.

To address this problem, following \cite{PT19a}, we can restrict the family of curves to be $\tau$-separated (they are piecwise-linear and critical points are a distance of at least $\tau$ to non-adjacent parts of the curve), and assume the landmark set is sufficiently dense (e.g., a grid with separation $\leq \tau/16$).  Under these conditions again $\dQ^{\sigma,p}$ is definite, and is a metric.  

In higher dimensions, our feature mapping and thus the distance $\dQ^{\sigma}$ could be extended to surfaces of codimension 1, given an appropriate generalization for dealing with surface boundary.

\paragraph{Computational complexity.}  For piecewise-linear curves $\gamma$ with $m$ segments, the feature map to $\R^n$ can be computed in $O(mn)$ time, by for each point $q_i$ taking the minimum distance among all segments.  Then distance between any pair of curves takes $O(n)$ time.  When $n$ is a constant (e.g., as $n=20$ recommended by \cite{PT19a}), then these runtimes are as fast as reading the data, and constant.  
In comparison, \Frechet distance takes time $O(m^2 \log m)$~\cite{AG95} and $\Omega(m^2 / \log m)$ assuming SETH~\cite{Bri14}.

\section{Signed Local Feature Size and Signed Medial Axis}

Given a curve $\gamma$ in $\R^2$, previous work studied ways it interacts with the ambient space.
The \emph{medial axis} $\MA(\gamma)$~\cite{lee1982medial,amenta2001power} is the set of points $q \in \R^2$ where the minimum distance $\min_{p \in \gamma} \|p-q\|$ is not realized by a unique point $p \in \gamma$.  
The \emph{local features size}~\cite{amenta1999surface} for a point $p \in \gamma$  defined by $\lfs_p(\gamma) = \inf_{r \in \MA(\gamma)} \|r-p\|$ is the minimum distance from $p$ to the medial axis of $\gamma$.  
We introduce signed versions of local feature size and medial axis which are intricately tied to the stability of $\dQ^{\sigma}$.  We use the notation ${\rm int}(\overline{pp'})$ to show the interior of a line segment $\overline{pp'}$, which is the line segment $\overline{pp'}$ without its endpoints.

\begin{figure}[h]
\includegraphics[width=1\textwidth]{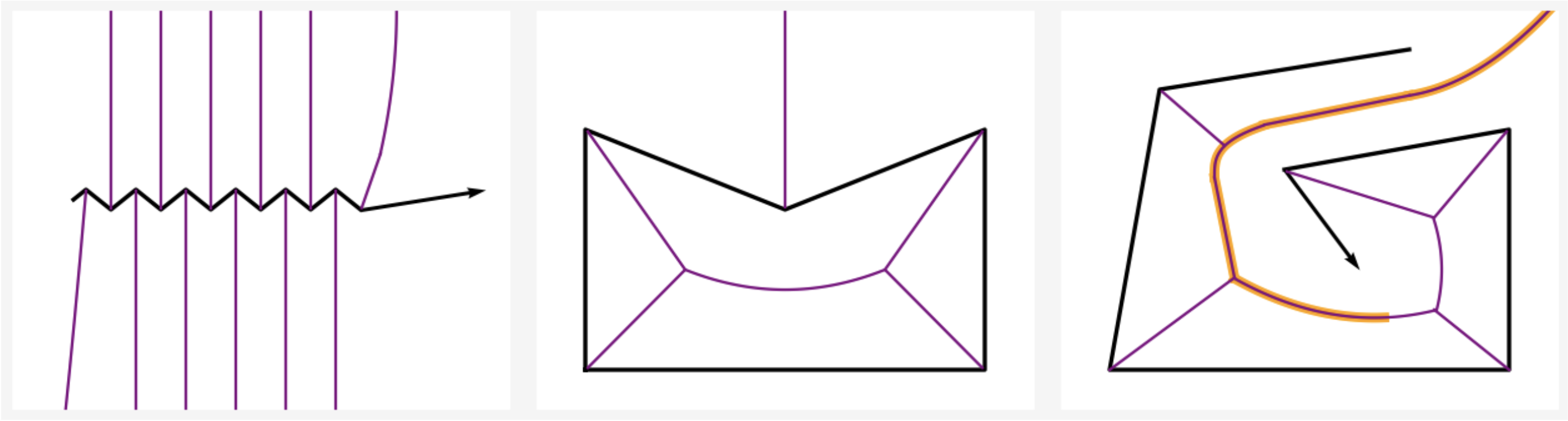}
\caption{\label{fig:MA-SMA}MA (in purple) and $\SMA$ (in orange) in example shapes (in black): noisy curve with large MA and no $\SMA$, closed curve with no $\SMA$, and spiral with $\SMA$ a subset of MA.}
\end{figure}

\begin{definition}[Signed Local Feature Size] \label{d3}
Let $\gamma \in \Gsim$ be a curve and $p$ be a point of  $\reg(\gamma)$. Define  
\[
\delta_p(\gamma) = \inf\{\|p-p'\|: \langle n_p, p-p' \rangle \langle n_{p'}, p'-p \rangle < 0, \ {\rm int}(\overline{pp'}) \cap \gamma = \emptyset, \ p' \in \reg(\gamma)\},
\]
where we assume that the infimum of the empty set is $\infty$. 
Then we introduce the {\it signed local feature size} ($\slfs$ in short) of $\gamma$ to be $\delta(\gamma)= \inf_{p\in \reg(\gamma)} \delta_p(\gamma)$. 
\end{definition}

\begin{exa}\label{e1}
Any line segment $\gamma$ has infinite $\slfs$ since ${\rm int}(\overline{pp'})$ is a subset of $\gamma$ for any $p,p' \in \gamma$.
Moreover, any simple closed curve in $\Gsim$ has infinite $\slfs$. Let $\gamma \in \Gsim$ be a closed simple curve which, w.l.o.g., oriented clockwise. The condition ${\rm int}(\overline{pp'}) \cap \gamma = \emptyset$ for any $p, p' \in \reg(\gamma)$ implies that the line segment $\overline{pp'}$ is completely inside $\gamma$ or completely outside $\gamma$. In both cases, $p-p'$ and $n_p$ (and similarly $p'-p$ and $n_{p'}$) will be on one side of the tangent line at $p$ (resp. $p'$) and so both inner products will be positive (since $n_p$ is perpendicular to the tangent line).
\end{exa}

The slfs captures twice the minimum distance between two parts of the curve, but keeping track of orientation, so the connecting segment has endpoints which emanates from the ``left'' side of one part of the curve and the ``right'' part of the other.  Thus this does not measure small wiggles in a curve, perhaps caused by noise (see Figure \ref{fig:MA-SMA}:Left), but only if a curve comes back on itself (Figure \ref{fig:MA-SMA}:Right).  
Unlike the (unsigned) local feature size, which captures how much perturbation is needed to change the homotopy, the slfs captures a safe distance a point can be from a curve so the notion of ``which side of the curve is it on" is well-defined.

With a simple closed curve, the notion of sidedness is unambiguous, but for non-closed curves this is not always easy to interpret.  Thus, next, we use this notation of signed local feature size to also adapt the related notion of signed medial axis. 
Unlike the (unsigned) medial axis, this does not measure the skeleton of a shape, but captures a transition boundary between curves, where the notion of sidedness changes.

For each $q \in \R^2$ and corresponding minDist point $p= \argmin_{p'\in \gamma} \|q-p'\|$ on $\gamma$, we need to define a normal direction $n_p(q)$ at $p$.  For regular points $p \in \reg(\gamma)$, this can be defined naturally by the right-hand rule.  For endpoints we use the normal vector of the tangent line compatible with the direction of the curve.  For non-endpoint critical point, there are technical conditions for non-simple curves (see Appendix \ref{app:alg}), but in general we use the direction $u$ which maximizes $|\langle u, q-p \rangle|$ with sign subject to the right-hard-rule.  While $p$ depends on $q$, for simplicity we avoid the notation $p(q)$. 

\begin{definition}[Signed Medial Axis] \label{d5}
Let $\gamma \in \Gsim$ and let $q$ be a point in $\R^2$. We say that $q$ belongs to the {\it signed medial axis} of $\gamma$ ($\SMA(\gamma)$ in short) if there are at least two points $p,p'$ on $\gamma$ such that $p, p'= \argmin_{p\in \gamma} \|q-p\|$ and $\langle n_p(q), p-p' \rangle \langle n_{p'}(q), p'-p \rangle < 0$. 
\end{definition}

The signed medial axis of a curve is a subset of its usual medial axis. Also, $\delta(\gamma) = \infty$ if and only if $\gamma$ has no signed medial axis, i.e. $\SMA(\gamma) = \emptyset$. Therefore, according to Example \ref{e1}, line segments and simple closed curves in $\Gamma$ have no signed medial axis. 
Figure \ref{fig:MA-SMA} also shows examples of both MA and $\SMA$.   Observe that it is possible to have $\SMA(\gamma) = \emptyset$ but $\MA(\gamma) \neq \emptyset$.
As a more intuitive example to see how to identify the $\SMA$, see the pictures in Figure \ref{Figslfs}, which provide examples of feature mapping $v_q$ for two curves where in the second complex picture we have labeled SMA and slfs.

\begin{figure}[h] 
\includegraphics[width=0.9 \textwidth]{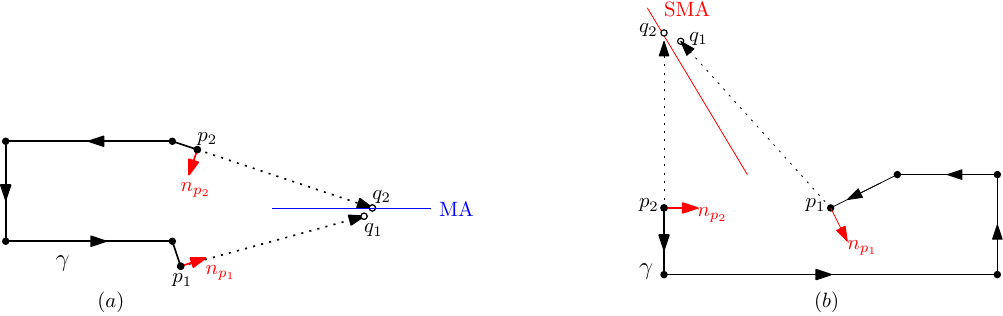}
\caption{\label{fig:q-endpoints}The cases that $q_1$ and $q_2$ choose different endpoints of $\gamma$ as $\argmin$ points}
\end{figure}

\section{Stability Properties of $\dQ^{\sigma, p}$}

In this section we proceed to the stability properties of the distance $\dQ^{\sigma, p}$.  
Our first goal is to show that $\dQ^{\sigma}$ is stable under perturbations of $Q$, which is given in two Theorems \ref{t3} and \ref{t4}. The next is to verify its stability under perturbations of curves (see Theorem \ref{t5}, its corollary and Theorem \ref{thm:eps})

\subsection{Stability of Landmarks $Q$}

As the new distance is contingent upon the choice of landmark points, it is important to examine its stability under variation of landmarks. Indeed, we would like to have some continuity-like properties in terms of perturbations of landmarks, which is  compatible with our intuition.

Before stating stability properties under perturbations of landmarks, we are going to discuss some cases that will not satisfy the desired inequality in Theorems \ref{t3} and \ref{t4}. As a result we will have to exclude these cases. The first case is when $\SMA(\gamma) = \emptyset$ and two landmarks $q_1$ and $q_2$ are on different sides of the medial axis of $\gamma$ and at least one of them chooses an endpoint as $\argmin$ point (see Figure \ref{fig:q-endpoints}(a)). The other case is when $\SMA(\gamma)$ is nonempty, $q_1$ and $q_2$ are in different sides of $\SMA(\gamma)$ and at least one of them chooses an endpoint as $\argmin$ point and is along the tangent of that endpoint (see Figure \ref{fig:q-endpoints}(b)). In both cases, $q_1$ can be arbitrarily close to $q_2$ but $|v_1(\gamma) - v_2(\gamma)|$ is likely to be roughly $1/\sqrt{e}$ since $v_2(\gamma)=0$. For instance, in Figure \ref{fig:q-endpoints}(a), $|v_1(\gamma) - v_2(\gamma)| = |v_1(\gamma)| = \frac{1}{\sigma} \|q_1-p_1\| e^{\frac{-\|q_1-p_1\|^2}{\sigma^2}}$, which can be as close as $1/\sqrt{e}$ when $\|q_1-p_1\|$ is about $\sigma/ \sqrt{2}$.

Summarizing these results:  If these end-point cases for $q_1, q_2$ do not occur (e.g., $\gamma$ is closed), then $|v_1(\gamma)  - v_2(\gamma)| \leq \frac{1}{\sigma} \|q_1- q_2\|$.  Or in the case of a $\SMA \neq \emptyset$, then $|v_1(\gamma)  - v_2(\gamma)| \leq \max\{\epsilon, \frac{1}{\sigma} \|q_1- q_2\|\}$ if $\sigma$ is set sufficiently small with respect to the slfs $\delta(\gamma)$. 

\begin{theorem}[Landmark stability I] \label{t3}
Let $\gamma \in \Gamma$ and $q_1, q_2$ be two points in $\mathbb{R}^2$. If $\delta(\gamma)= \infty$ and $q_1$ and $q_2$ do not satisfy the above first case (e.g., $\gamma$ is a closed curve), then $|v_1(\gamma)- v_2(\gamma)|\leq \frac{1}{\sigma}\|q_1-q_2\|$.
\end{theorem}
\begin{proof}
Let $p_1= \argmin_{p\in \gamma} \|q_1-p\|$ and $p_2 = \argmin_{p\in \gamma} \|q_2-p\|$. We prove the theorem in four cases. \\
{\bf Case 1.} 
$v_1(\gamma) v_2(\gamma) \leq 0$, the line segment $\overline{q_1q_2}$ passes through $\gamma$ and  $p_1$ and $p_2$ are not endpoints (see Figure \ref{fig:q-cases}(a)). 
Let $p$ be the intersection of the segment $\overline{q_1q_2}$ with $\gamma$. Then
\[
\begin{array}{ll}
|v_1(\gamma)- v_2(\gamma)| \!\!\! & 
= \bigg| \frac{1}{\sigma}\langle q_1-p_1, n_{p_1}(q_1)\rangle e^{- \frac{\|q_1-p_1\|^2}{\sigma^2}} - \frac{1}{\sigma}\langle q_2-p_2, n_{p_2}(q_2)\rangle e^{- \frac{\|q_2-p_2\|^2}{\sigma^2}} \bigg| \vspace{0.1cm} \\ &
\leq \frac{1}{\sigma} \|q_1-p_1\| e^{- \frac{\|q_1-p_1\|^2}{\sigma^2}} + \frac{1}{\sigma} \|q_2-p_2\| e^{- \frac{\|q_2-p_2\|^2}{\sigma^2}} \vspace{0.2cm} \\ &
\leq \frac{1}{\sigma} (\|q_1-p_1\| + \|q_2-p_2\|) 
\leq \frac{1}{\sigma} (\|q_1-p\| + \|q_2-p\|) 
= \frac{1}{\sigma} \|q_1-q_2\|.
\end{array}
\]
{\bf Case 2.} 
$v_1(\gamma) v_2(\gamma) \geq 0$, the line segment $\overline{q_1q_2}$ does not pass through $\gamma$ and  $p_1$ and $p_2$ are not endpoints (see Figure \ref{fig:q-cases}(b)). 
Without loss of generality we may assume that both $v_1(\gamma)$ and $v_2(\gamma)$ are non-negative. In this case, $q_1-p_1$ and $q_2-p_2$ are parallel to $n_{p_1}(q_1)$ and $n_{p_2}(q_2)$ respectively. Therefore, $\langle q_1-p_1, n_{p_1}(q_1)\rangle= \|q_1-p_1\|$ and $\langle q_2-p_2, n_{p_2}(q_2)\rangle= \|q_2-p_2\|$. Utilizing the fact that the function $f(x)= \frac{x}{\sigma} e^{-x^2/ \sigma^2}$ is Lipschitz with constant $\frac{1}{\sigma}$, we get 
\[
|v_1(\gamma)- v_2(\gamma)| \leq \frac{1}{\sigma} |\|q_1-p_1\| - \|q_2-p_2\||.
\]
Now applying triangle inequality we infer  $\|q_1-p_1\| \leq \|q_1-p_2\| \leq \|q_1-q_2\| + \|q_2-p_2\|$ and so by symmetry, $|\|q_1-p_1\| - \|q_2-p_2\|| \leq \|q_1-q_2\|$. Therefore, $|v_1(\gamma)- v_2(\gamma)| \leq \frac{1}{\sigma} \|q_1-q_2\|$. 
\\
{\bf Case 3.} Endpoints. 
Let $\ell$ be the tangent line at an endpoint $p$ on $\gamma$, $n_p$ be its unique unit normal vector and let $q_1$ and $q_2$ be in different sides of $\ell$ and $p_1 = p_2 = p$ (see Figure \ref{fig:q-cases}(c)). Assume $q$ is the intersection of the segment $\overline{q_1q_2}$ with $\ell$. Then $\langle n_p, q-p \rangle = 0$ and so noting that $n_{p_1}(q_1)=n_{p_2}(q_2)=n_p$ we have $\langle q_1-p, n_p \rangle = \langle q_1-q, n_p \rangle$ and $\langle q_2-p, n_p \rangle = \langle q_2-q, n_p \rangle$. Therefore, 
\[
\begin{array}{ll}
|v_1(\gamma)- v_2(\gamma)| & 
= \bigg| \frac{1}{\sigma} \langle n_{p}, \frac{q_1-p}{\|q_1-p\|} \rangle \|q_1\|_{\infty,p} \, e^{\frac{\|q_1-p\|^2}{\sigma^2}} - \frac{1}{\sigma} \langle n_{p}, \frac{q_2-p}{\|q_2-p\|} \rangle \|q_2\|_{\infty,p} \, e^{\frac{\|q_2-p\|^2}{\sigma^2}} \bigg| \vspace{0.2cm} \\ &
= \frac{1}{\sigma} \bigg| \Big\langle n_p, \frac{\|q_1\|_{\infty,p}}{\|q_1\|_{2,p}} e^{- \frac{\|q_1-p\|^2}{\sigma^2}} (q_1-q)- \frac{\|q_2\|_{\infty,p}}{\|q_2\|_{2,p}} e^{- \frac{\|q_2-p\|^2}{\sigma^2}} (q_2-q) \Big\rangle \bigg| \vspace{0.2cm} \\ &
= \frac{1}{\sigma} \bigg\| \frac{\|q_1\|_{\infty,p}}{\|q_1\|_{2,p}} e^{- \frac{\|q_1-p\|^2}{\sigma^2}} (q_1-q)- \frac{\|q_2\|_{\infty,p}}{\|q_2\|_{2,p}} e^{- \frac{\|q_2-p\|^2}{\sigma^2}} (q_2-q) \bigg\| \vspace{0.2cm} \\ &
\leq \frac{1}{\sigma} (\|q_1-q\| + \|q_2-q\|) 
= \frac{1}{\sigma} \|q_1-q_2\|.
\end{array}
\]
If $q_1$ and $q_2$ are in one side of $\ell$ and $p_1 = p_2 = p$, the proof is the same as in Case 2. We only need to apply Cauchy-Schwarz inequality. 
\\
{\bf Case 4.} 
The case where $p_1$ is an endpoint but $p_2$ is not can be gained from a combination of above cases. Basically, choose a point $q$ on the line segment $\overline{q_1q_2}$ so that $q-p_1$ is parallel to $n_{p_1}$ and then use the triangle inequality. 
\end{proof}

\begin{figure}[t] 
\includegraphics{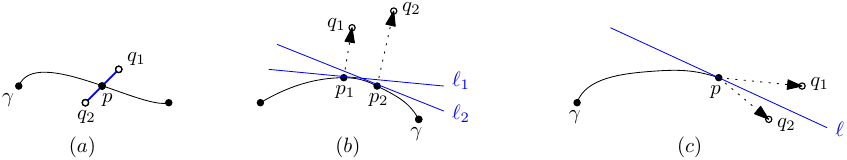}
\caption{\label{fig:q-cases}$q_1$ and $q_2$ in different cases}
\end{figure}

\begin{theorem}[Landmark stability II] \label{t4}
Let $\gamma \in \Gsim$ and $q_1, q_2$ be two points in $\mathbb{R}^2$ not satisfying the second case mentioned before Theorem \ref{t3}. If $\delta(\gamma)< \infty$, $\epsilon \leq \frac{\delta(\gamma)}{4}$ is an arbitrary positive real number and $\sigma \leq \delta(\gamma) / (4(1+ \sqrt{\ln(2/\epsilon)}))$, then \vspace{-1mm}
\[
|v_1(\gamma)- v_2(\gamma)|\leq \max\{\epsilon, \frac{1}{\sigma}\|q_1-q_2\|\}.
\]
\end{theorem}
\begin{proof}
By Theorem \ref{t3} it is enough to consider only the case where there is a signed medial axes, say $M$, and $q_1$ and $q_2$ are in different sides of $M$ (the case they are in same side of $M$ is included in Theorem \ref{t3}). For the sake of convenience assume $x= \|q_1-p_1\|$, $y= \|q_2-p_2\|$ and let $z$ be either $x$ or $y$. The proof is based on the following observations. 
\begin{itemize} 
\item[{(O1)}] If $z \leq \frac{\delta(\gamma)}{4}$, then $\frac{z}{\sigma} e^{-\frac{z^2}{\sigma^2}} \leq \frac{z}{\sigma} \leq \frac{\delta(\gamma)}{4 \sigma}$. 
\item[{(O2)}] If $z \geq \frac{\delta(\gamma)}{4}$, then $\frac{z}{\sigma} \geq \frac{\delta(\gamma)}{4 \sigma} \geq 1+ \sqrt{\ln(2/\epsilon)}$. Hence, Employing the inequality $\frac{z}{\sigma} e^{-\frac{z^2}{\sigma^2}} \leq e^{\frac{2z}{\sigma}-1} e^{-\frac{z^2}{\sigma^2}}= e^{-(\frac{z}{\sigma}-1)^2}$ we get $\frac{z}{\sigma} e^{-\frac{z^2}{\sigma^2}} \leq \frac{\epsilon}{2}$. 
\end{itemize}
Now if $\|q_1-q_2\| \leq \epsilon$, noting that $\epsilon \leq \frac{\delta(\gamma)}{4}$, we have $x, y \geq \frac{3 \delta(\gamma)}{4} \geq \frac{\delta(\gamma)}{4}$. Thus, by (O2), we get $|v_1(\gamma)- v_2(\gamma)|= \dfrac{x}{\sigma} e^{-\frac{x^2}{\sigma^2}} + \dfrac{y}{\sigma} e^{-\frac{y^2}{\sigma^2}} \leq \epsilon$. \\
Otherwise, $\|q_1-q_2\| \geq \epsilon$ and we encounter four cases (see Figure \ref{fig:SMA-case} (Left)). 
\\
{\bf Case 1.}
If $x, y \leq \frac{\delta(\gamma)}{4}$, then by (O1), 
$|v_1(\gamma)- v_2(\gamma)| = \frac{x}{\sigma} e^{-\frac{x^2}{\sigma^2}} + \frac{y}{\sigma} e^{-\frac{y^2}{\sigma^2}} \leq \frac{\delta(\gamma)}{2 \sigma} \leq \frac{1}{3 \sigma} \|q_1-q_2\|$. \\
{\bf Case 2.}
If $x\geq \frac{\delta(\gamma)}{4}$ and $y \leq \frac{\delta(\gamma)}{4}$, then applying (O1) and (O2) we infer 

$|v_1(\gamma)- v_2(\gamma)| = \dfrac{x}{\sigma} e^{-\frac{x^2}{\sigma^2}} + \dfrac{y}{\sigma} e^{-\frac{y^2}{\sigma^2}} \leq  \frac{\epsilon}{2} + \frac{\delta(\gamma)}{4 \sigma}\leq \Big(1+\sqrt{\ln(\frac{2}{\epsilon})} \, \Big)+\frac{\delta(\gamma)}{4 \sigma} \leq \frac{\delta(\gamma)}{2 \sigma} \leq  \frac{2}{3 \sigma} \|q_1-q_2\|.$ \\
{\bf Case 3.} 
The case $x \leq \frac{\delta(\gamma)}{4}$ and $y \geq \frac{\delta(\gamma)}{4}$ is the same as Case 2. \\
{\bf Case 4.}
Finally, if $x\geq \frac{\delta(\gamma)}{4}$ and $y \geq \frac{\delta(\gamma)}{4}$, by (O2),  
$|v_1(\gamma)- v_2(\gamma)| = \frac{x}{\sigma} e^{-\frac{x^2}{\sigma^2}} + \frac{y}{\sigma} e^{-\frac{y^2}{\sigma^2}} \leq \epsilon$. 
\end{proof}

\begin{figure}[t]
\includegraphics[width=\textwidth]{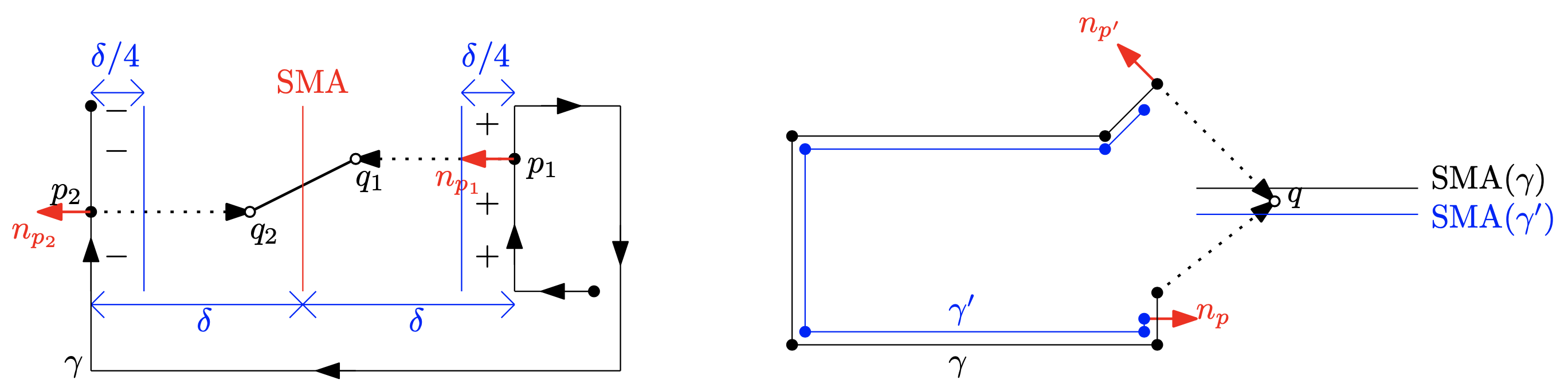}
\caption{\label{fig:SMA-case}Left: The case that there is a $\SMA$ and $q_1$ and $q_2$ in different sides of $\SMA(\gamma)$. Right: The case that a landmark point $q$ lies between $\SMA(\gamma)$ and $\SMA(\gamma')$.}
\end{figure}

\subsection{Stability of Curves} \label{sec:curve-stab}

We next would like to have a continuity-like property for $\dQ^{\sigma}$ in terms of perturbations of curves. This property considers similar curves, where one is a small perturbation of the other, and shows that these curves will necessarily also have a small $\dQ^{\sigma}$ distance.  In particular, we define the scale of these perturbations in terms of the \Frechet and Hausdorff distance; hence this shows consistency between $\dQ^{\sigma}$ and these distances, at least at small scales.  At large scales, and for certain classes of curves, such a result is not possible~\cite{indyk1998approximate,DS17}, so we need to condition the classes of curves for which it holds.  

Moreover, non-closed curves create subtle issues around endpoints.  
The example in Figure \ref{fig:SMA-case} (Right) shows that, without controlling behavior of endpoints, we may make $\gamma$ arbitrarily close to $\gamma'$ in \Frechet distance, whereas $\dQ^{\sigma}(\gamma, \gamma')$ is possibly $1/\sqrt{e}$, where $Q = \{q\}$. This is the case where $q$ lies between $\SMA(\gamma)$ and $\SMA(\gamma')$. So we cannot get the desired inequality ($\dQ^\sigma(\gamma,\gamma') \leq \dF(\gamma,\gamma')$) for this case. 
Thus, the landmarks which fall between the signed medial axes may cause otherwise similar curves to have different signatures.  For a large domain $\Omega$ (and especially with $\sigma$, relatively small) these should be rare, and then $\dQ^\sigma$ which averages over these landmarks should not be majorly effected.  
We formalize when this is the case in the next theorem, and its corollary which shows that if two curves $\gamma,\gamma'$ are closed (Corollary \ref{c4}), then $\dQ^\sigma(\gamma,\gamma') \leq \frac{1}{\sigma}\dF(\gamma,\gamma')$. We also obtain a $(1+\eps)$-relative error inequality $\dQ^{\sigma}(\gamma, \gamma') \leq \frac{1}{\sigma} (1+\eps) \dF(\gamma, \gamma')$ in Theorem \ref{thm:eps} when $\gamma,\gamma'$ are not necessarily closed. 

\begin{theorem}[Stability under \Frechet perturbation of curves] \label{t5}
Let $\gamma, \gamma' \in \Gsim$ and $q_i \in \mathbb{R}^2$. If one of the following three conditions hold, then  
$|v_i(\gamma)- v_i(\gamma')|\leq \frac{1}{\sigma} \dF(\gamma, \gamma')$. \\
(1) $v_i(\gamma)v_i(\gamma') \geq 0$; \\
(2) $v_i(\gamma)v_i(\gamma') \leq 0$ and $q_i$ is on a line segment $\overline{\gamma(t) \gamma'(\alpha(t))}$, for some $t$, of the alignment between $\gamma$ and $\gamma'$ achieving the optimal \Frechet distance;  \\
(3) $q_i$ is far enough from both curves: $\min_{p \in \gamma} \|q_i - p\|, \min_{p' \in \gamma'} \|q_i - p'\| \geq \sigma (1 + \sqrt{\ln(2\sigma / \dF(\gamma,\gamma'))})$. 
\end{theorem}
\begin{proof}
Let $p= \argmin_{p\in \gamma} \|q_i-p\|$ and $p'= \argmin_{p\in \gamma'} \|q_i-p\|$. 

(1) Let $p=\gamma(t)$ and $p'= \gamma(t')$ for some $t,t' \in [0,1]$ and without loss of generality assume that $t \leq t'$. Then $\|q_i-p'\| \leq \|q_i- \gamma'(t)\|$ and so 
\[
\quad\qquad \|q_i-p'\| - \|q_i-p\| \leq \|q_i- \gamma'(t)\| - \|q_i-p\| \leq \|p- \gamma'(t)\|= \|\gamma(t)- \gamma'(t)\| \leq \|\gamma - \gamma'\|_{\infty}.
\]
Similarly, $\|q_i-p\| - \|q_i-p'\| \leq \|\gamma - \gamma'\|_{\infty}$. 
Now using the fact that the function $f(x)= \frac{x}{\sigma} e^{-\frac{x^2}{\sigma^2}}$ is Lipschitz, considering $l^{\infty}$-norm at endpoints, we get 
\[
|v_i(\gamma)- v_i(\gamma')|\leq \frac{1}{\sigma} |\|q_i-p'\| - \|q_i-p\|| \leq \frac{1}{\sigma} \|\gamma - \gamma'\|_{\infty}.
\]
This shows that for two arbitrary reparametrizations $\alpha$ and $\alpha'$ of $[0,1]$ we have 
$
|v_i(\gamma \circ \alpha)- v_i(\gamma' \circ \alpha')|\leq  \frac{1}{\sigma} \|\gamma \circ \alpha - \gamma' \circ \alpha'\|_{\infty}.
$
Noting that a reparametrization of a curve does not change either the range or the direction of the curve, we get $|v_i(\gamma \circ \alpha)- v_i(\gamma'\circ \alpha')|= |v_i(\gamma)- v_i(\gamma')|$. Thus taking the infimum over all reparametrizations $\alpha$ and $\alpha'$ we obtain $|v_i(\gamma)- v_i(\gamma')| \leq \frac{1}{\sigma} \dF(\gamma, \gamma')$. 

(2) Let $r = \dF(\gamma, \gamma')$ and let $q_i$ be on a line segment alignment of $\gamma, \gamma'$ with length at most $r$. So, there are points $a$ and $b$ on $\gamma$ and $\gamma'$ respectively within distance $r$ such that $q_i$ lies on $\overline{ab}$. Hence we have 
$\|p-q_i\| + \|q_i-p'\| \leq \|a-q_i\| + \|b-q_i\| = \|a-b\| \leq r = \dF(\gamma, \gamma')$, and thus 
\[
|v_i(\gamma)- v_i(\gamma')| \leq \frac{1}{\sigma} (\|p-q_i\| + \|q_i -p'\|) \leq \frac{1}{\sigma} \dF(\gamma, \gamma').
\]

(3) This case implies $|v_i(\gamma)|, |v_i(\gamma')| \leq \frac{1}{2\sigma}\dF(\gamma,\gamma')$.   
\end{proof}

\begin{cor} \label{c4}
Let $\gamma, \gamma' \in \Gsim$ be closed curves with both oriented clockwise/counterclockwise. Then $\dQ^{\sigma}(\gamma, \gamma') \leq \frac{1}{\sigma} \dF(\gamma, \gamma')$. 
\end{cor}
\begin{proof}
Using Condition (1) of Theorem \ref{t5}, it is enough to show that if $q \in A \setminus A'$, then $|v_q(\gamma) - v_q(\gamma')| \leq \frac{1}{\sigma} \dF(\gamma, \gamma')$, where $A, A'$, by Jordan's curve theorem, are the regions bounded by $\gamma, \gamma'$ respectively. The case $q \in A' \setminus A$ comes by symmetry. We claim that $q$ lies on a line segment of the alignment between $\gamma$ and $\gamma'$ achieving the optimal Fr$\rm\acute{e}$chet distance. Then the statement of the theorem follows by Theorem \ref{t5} as Condition (2) will hold true.  

Assume to the contrary that $q$ does not lie on such a line segment. Now consider $\gamma$ and $\gamma'$ in $\mathbb{R}^2 \setminus \{q\}$ and let $x_0$ be a point in $\mathbb{R}^2 \setminus \{q\}$. Using the fact that in path-connected topological spaces, like $\mathbb{R}^2 \setminus \{q\}$, up to isomorphism of groups the fundamental group of the space is independent of the choice of base point $x_0$ (see~\cite{hatcher2005algebraic} for example), we see that in $\mathbb{R}^2 \setminus \{q\}$, $\gamma'$ is contractible (i.e. homotopic to $x_0$) but $\gamma$ is homotopic to $\mathbb{S}^1$, the unit circle. It means that there cannot be a homotopy between $\gamma$ and $\gamma'$ since the homotopy relation is an equivalence relation and $\mathbb{S}^1$ is not contractible, as their fundamental groups are not isomorphic. 

On the other hand, let $r = \dF(\gamma, \gamma')$ and let $\alpha$ be a reparametrization achieving the optimal alignment for Fr$\rm\acute{e}$chet distance. It means that the mapping $F:[0,1] \times [0,1] \to \mathbb{R}^2$ defined by $F(s,t) = s (\gamma \circ \alpha)(t) + (1-s) \gamma'(t)$ is a straight line homotopy between $\gamma\circ \alpha$ and $\gamma'$, and by assumption on $q$ we know that this homotopy occurs in $\mathbb{R}^2 \setminus \{q\}$.  Since reparametrizations do not change the homotopy class of curves, we observe that there is a homotopy between $\gamma$ and $\gamma'$ in $\mathbb{R}^2 \setminus \{q\}$, which is a contradiction. 
\end{proof}

\begin{theorem} \label{thm:eps}
Let $\gamma, \gamma' \in \Gsim$, $r = \dF(\gamma, \gamma') > 0$, $\eps >0$ and $\sigma \leq \eps^2 r^2/2$. Then 
\[
\dQ^{\sigma}(\gamma, \gamma') \leq \frac{1}{\sigma} (1+\eps) \dF(\gamma, \gamma'). 
\]
\end{theorem}
\begin{proof} 
Let $q_i \in Q$, $p= \argmin_{p\in \gamma} \|q_i-p\|$ and $p'= \argmin_{p\in \gamma'} \|q_i-p\|$. If $\|q_i-p\| \leq \eps r/2$, then considering the fact that $\gamma'$ lives in the $r$-neighborhood of $\gamma$ we have $\|q_i-p'\| \leq \eps r/2 + r$. Hence, $|v_i(\gamma)| \leq \frac{1}{\sigma} \|q_i-p\| \leq \frac{\eps/2}{\sigma} r$, and similarly, $|v_i(\gamma')| \leq \frac{1}{\sigma} \|q_i-p' \| \leq \frac{(1+\eps/2)}{\sigma} r$. Consequently, $|v_i(\gamma) - v_i(\gamma')| \leq \frac{1}{\sigma} (1+\eps) \dF(\gamma, \gamma')$. The case $\|q_i-p' \| \leq \eps r/2$ comes by symmetry. 

On the other hand, let $\|q_i-p\| > \eps r/2$ and $\|q_i-p' \| > \eps r/2$. Then using the inequality $\sigma \leq \eps^2 r^2/2$ along with the fact that the function $f(x) = \frac{x}{\sigma} e^{-x^2/\sigma^2}$ is decreasing on the interval $[\sqrt{\sigma/2}, \infty)$, we gain 
\[
|v_i(\gamma) - v_i(\gamma')| \leq |v_i(\gamma)| + |v_i(\gamma')| \leq f(\|q_i-p\|) + f(\|q_i-p'\|) \leq 2 f(\eps r/2) \leq \frac{\eps}{\sigma} r.
\]
This completes the proof. 
\end{proof}

\subsection{Interleaving Bounds for $l^\infty$ Variants}
\label{sec:interleaving}

Using the $l^\infty$ variants, we can show a stronger interleaving property.  
Let $\Omega$ be a bounded domain in $\mathbb{R}^2$. Let $\diam(\Omega) = \sup_{x,y \in \Omega} \|x-y\|$ be the diameter of $\Omega$. We also denote by $\Gsim_{\Omega}$ the subset of $\Gsim$ containing all curves with image in $\Omega$.  
In Section \ref{app:Hausdorff} we show if $\gamma,\gamma' \in \Gamma_\Omega$ and $Q$ is uniform on a domain $\Omega$, then $\dH(\gamma,\gamma') = {\dQmD}^{,\infty}(\gamma,\gamma')$.  
The signed variant $\dQ^{\sigma, \infty}$ is more related to $\dF$, but it is difficult to show an interleaving result in general because if a curve cycles around multiple times, its image may not significantly change, but its \Frechet distance does.  However, by appealing to a connection to the Hausdorff distance, and then restricting to closed and convex (Corollary \ref{c6}) or $\kappa$-bounded~\cite{AKW04} curves (Corollary \ref{c7}), we can still achieve interleaving bounds. 

We first focus on closed curves, so $\slfs$ is infinite, and there are no boundary issues; thus it is best to set $\sigma$ sufficiently large so the $\exp(-\|p-q\|^2/\sigma^2)$ term in $v_q$ goes to $1$ and can be ignored. Regardless, $\dQ^{\sigma} \leq 1/\sqrt{2e}$.  Note that $v_q$ and hence $\dQ$ has a $\frac{1}{\sigma}$ factor, so those terms in the expressions cancel out.  

\begin{lemma} \label{lem:H-lb}
Assume that $Q$ is a uniform measure on $\Omega$ and $\sigma$ is sufficiently large. Let $\gamma, \gamma' \in \Gsim_{\Omega}$ be two closed curves such that $\dH(\gamma, \gamma') \leq \frac{\sigma}{\sqrt{2e}}$.  Then $\frac{1}{\sigma} \dH(\gamma, \gamma') \leq \dQ^{\sigma, \infty}(\gamma, \gamma')$. 
\end{lemma}
\begin{proof}
Let $r = \dH(\gamma, \gamma')$. Without loss of generality we can assume that $r = \sup_{p \in \gamma} \inf_{p' \in \gamma'} \|p - p'\|$. Since the range of $\gamma$ is compact (the image of a compact set under a continuous map is compact), there is $p \in \gamma$ such that $r = \min_{p' \in \gamma'} \|p - p'\|$. Similarly, by the continuity of the range of $\gamma'$ we conclude that there is $p' \in \gamma$ such that $r = \|p - p'\|$. Because $Q$ is dense in $\Omega$ then $p \in Q$. Since $p' = \argmin_{p'' \in \gamma'} \|p - p''\|$, we observe that $v_p(\gamma') = \frac{1}{\sigma} \|p-p'\| = \frac{r}{\sigma}$. On the other hand, $v_p(\gamma) = 0$ as $p \in \gamma$. Therefore, at least one of the components of the sketched vector $v_Q(\gamma') - v_Q(\gamma)$ is $\frac{r}{\sigma}$ and so $\sigma \dQ^{\sigma, \infty}(\gamma, \gamma') \geq r$. 
\end{proof}

\begin{cor}\label{c1}
Assume $Q$ is a uniform measure on $\Omega$ and $\sigma$ is sufficiently large. Let $\gamma, \gamma' \in \Gsim_{\Omega}$ be two closed curves. Then $\dH(\gamma, \gamma') \leq \sqrt{2e} \, {\rm diam}(\Omega) \dQ^{\sigma, \infty}(\gamma, \gamma')$.
\end{cor}

\begin{cor} \label{c6}
Let $Q$ be a uniform measure on $\Omega$ and $\sigma$ be sufficiently large. Let $\gamma, \gamma' \in \Gsim_{\Omega}$ be two closed convex curves with both oriented clockwise/counterclockwise and $\dH(\gamma, \gamma') \leq \frac{\sigma}{\sqrt{2e}}$. Then $\dQ^{\sigma, \infty}(\gamma, \gamma') = \frac{1}{\sigma} \dF(\gamma, \gamma')$. 
\end{cor}
\begin{proof}
Applying Lemma \ref{lem:H-lb} and the $p = \infty$ version of Corollary \ref{c4} we get 
$\frac{1}{\sigma} \dH(\gamma, \gamma') \leq \dQ^{\sigma, \infty}(\gamma, \gamma') \leq \frac{1}{\sigma} \dF(\gamma, \gamma')$. 
Now by Theorem 1 of \cite{AKW04} we know that the Hausdorff and Fr$\rm{\acute{e}}$chet distances coincide for closed convex curves. Therefore, $\dH(\gamma, \gamma') = \dF(\gamma, \gamma')$ and the proof is complete.
\end{proof}

The proof of Lemma \ref{lem:H-lb} shows that the inequality $\frac{1}{\sigma}\dH(\gamma, \gamma') \leq  \dQ^{\sigma, \infty}(\gamma, \gamma')$ remains valid for non-closed curves $\gamma$ and $\gamma'$ as long as $p'$ in the proof is not an endpoint of $\gamma'$. 
A piecewise linear curve $\gamma$ in $\mathbb{R}^2$ is called {\it $\kappa$-bounded}~\cite{AKW04} for some constant $\kappa \geq 1$ if for any $t,t' \in [0,1]$ with $t < t'$, $p = \gamma(t)$, $p' = \gamma(t')$ we have 
$\gamma([t,t']) \subseteq  B_r(p) \cup B_r(p')$,
where $r=\frac{\kappa}{2} \|p-p' \|$. 
The class of $\kappa$-bounded curves comprises of $\kappa$-straight curves \cite{AKW04},  curves with increasing chords \cite{Rot94} and self-approaching curves \cite{AAIKLR01}. 

\begin{cor} \label{c7}
Let $Q$ be a uniform measure on $\Omega$ and $\sigma$ be sufficiently large. Let $\gamma',\gamma \in \Gsim_{\Omega}$ be $\kappa$-bounded, with the Hausdorff distance not achieved at endpoints, and $\dH(\gamma, \gamma') \leq \frac{\sigma}{\sqrt{2e}}$. Then 
\[
\dF(\gamma, \gamma') \leq \sigma(\kappa+1) \dQ^{\sigma, \infty}(\gamma, \gamma'). 
\]
\end{cor}
\begin{proof}
Using the non-closed version of Lemma \ref{lem:H-lb} we have $\dH(\gamma, \gamma') \leq \sigma \dQ^{\sigma, \infty}(\gamma, \gamma')$. Now, since $\gamma$ and $\gamma'$ are $\kappa$-bounded, by Theorem 2 of \cite{AKW04} we have $\dF(\gamma, \gamma') \leq (\kappa +1) \dH(\gamma, \gamma')$. Combining these inequalities we get the desired result. 
\end{proof}

\subsection{Relation of $\dQmD$ to the Hausdorff distance}
\label{app:Hausdorff}

In this subsection we show that the unsigned variant of the sketch $v_{q_i}^{\mathsf{mD}}(\gamma)$ based only on the minDist function, has a strong relationship to the Hausdorff distance.  In particular, when $Q$ is dense enough, and the $l^\infty$ variant is used, they are identical.  

\begin{theorem} \label{old d_Q}
Let $\gamma, \gamma'$ be two continuous curves and $q \in \mathbb{R}^2$. Then 
$|v_q^{\mathsf{mD}}(\gamma) - v_q^{\mathsf{mD}}(\gamma') | \leq \dH(\gamma, \gamma').$
Consequently, $\dQmD(\gamma, \gamma') \leq \dH(\gamma, \gamma')$ for any landmark set $Q$. 
\end{theorem} 
\begin{proof}
Let $r = \dH(\gamma, \gamma')$. Suppose $p = \argmin_{p\in \gamma} \|q-p\|$ and $p' = \argmin_{p'\in \gamma'} \|q-p'\|$. Let also $y = \argmin_{y\in \gamma} \|y-p' \|$ and $y' = \argmin_{y' \in \gamma} \|y'-p \|$. Then we have $\|q-p\| \leq \|q-y\|$ and $\|q-p' \| \leq \|q-y' \|$ and according to the definition of the Hausdorff distance $\|y-p' \| \leq r$ and $\|y' - p\| \leq r$. Now there are two possible cases: 

(i) $\|q-p\| \leq \|q-p' \|$. Then using the triangle inequality we get $0 \leq \|q-p\| - \|q-p' \| \leq \|q-y\| - \|q-p' \| \leq \|y- p' \| \leq r$.

(ii) $\|q-p' \| \leq \|q-p\|$. Then $0 \leq \|q-p' \| - \|q-p\| \leq \|q-y' \| - \|q-p\| \leq \|y' - p\| \leq r$. \\
Therefore, $| \|q-p\| - \|q-p' \| | \leq r$. The next inequality is immediate as we take average in computing $\dQ$. 
\end{proof}

\begin{cor}\label{c2}
Let $\Omega \subset \mathbb{R}^2$ be a bounded domain and $Q$ is dense in $\Omega$. If the range of $\gamma, \gamma'$ are included in $\Omega$, then ${\dQmD}^{\infty}(\gamma, \gamma') = \dH(\gamma, \gamma')$. 
\end{cor}
\begin{proof}
Employing Theorem \ref{old d_Q} we only need to show ${\dQmD}^{\infty}(\gamma, \gamma')  \geq \dH(\gamma, \gamma')$. Let $r = \dH(\gamma, \gamma')$. Without loss of generality we can assume that $r = \sup_{p \in \gamma} \min_{p' \in \gamma'} \|p-p'\|$.  Since the range of $\gamma$ is compact (the image of a compact set under a continuous map is compact), there is $p \in \gamma$ such that $r = \min_{p' \in \gamma'} \|p-p'\|$. Similarly, by continuity of the range of $\gamma'$ we conclude that there is $p' \in \gamma$ such that $r = \|p - p'\|$. Because $Q$ is dense in $\Omega$, without loss of generality, with an $\varepsilon$-discussion, we may assume that $p \in Q$. Since $p' = \argmin_{p' \in \gamma'} \|p - p' \|$, we observe that $v_p^{\mathsf{mD}}(\gamma') = \|p-p' \| = r$. On the other hand, $v_p^{\mathsf{mD}}(\gamma) = 0$ as $p \in \gamma$. Therefore, at least one of the components of the sketched vector $v_Q^{\mathsf{mD}}(\gamma') - v_Q^{\mathsf{mD}}(\gamma)$ is $r$ and so ${\dQmD}^{\infty}(\gamma, \gamma') \geq r$. 
\end{proof}

\section{Experiments: Trajectories Analysis via $\dQ^{\sigma}$ Distance} \label{S4}

Like with recent vectorized distance $\dQmD$~\cite{PT19a,PT20}, this structure allows for very simple and powerful data analysis.  Nearest neighbor search can use heavily optimized libraries~\cite{KGraph,FALCONN}. Clustering can use Lloyd's algorithm for $k$-means clustering.  

Here we compare the use of $\dQ^\sigma$ with $\dQmD$ for various classification tasks.  Previous work~\cite{PT19a} demonstrated that $\dQmD$ performed comparably or significantly better than KNN classifiers with a large set of other distances (Dynamic Time Warping (DTW), discrete Hausdorff, discrete \Frechet (dF), LSH approximations of it, Edit distance with Replacements, LCSS).  We compare with DTW from \texttt{tslearn} and dF from \texttt{similaritymeasures} python libraries.
These other methods, can only use KNN classifiers, but the vectorized representation for $\dQ^\sigma$ and $\dQmD$ allow them to use any classification technique.  
We found $\dQmD$ and $\dQ^\sigma$ perform similarly on most data sets.  Here we demonstrate the efficacy of $\dQ^\sigma$ on three real-world data set and one synthetic one.  
When the orientation is inessential, it typically performs similar to $\dQmD$, but where orientation is essential, $\dQ^\sigma$ provides significant advantage over $\dQmD$. 
In these comparisons we choose $|Q| = 20$ landmarks for each, chosen randomly near the data points, and listed in the Appendix.  
The trajectories are randomly split (70/30) into train and test data; we report test errors (misclassification rates) for several classifiers, averaged over $1000$ random test-train splits in Table \ref{table-Beijing-Synthetic}.

\paragraph{Car-Bus.}
We first consider a set of 78 car and 45 bus trajectories from the GPS Trajectories Data Set \cite{GTDS2016} (available for public in UCI Machine Learning Repository) recorded in Aracuja, Brazil, after removing ones with $< 10$ critical points; see Figure \ref{car-bus-pigeon} (Left) in Appendix \ref{app:Q}. We randomly generate 20 points as landmarks $Q$ around the curves. 
With the $v_Q^\sigma$ we achieve the best mis-classification rate of $0.1567$ with Random Forest, while $v_Q^{\mathsf{mD}}$ can achieve error only $0.1830$.  Other non-linear classifiers achieve error between $0.22$ and $0.25$ for $v^\sigma_Q$  and $0.21$ and $0.30$ for $v_Q^{\mathsf{mD}}$.  The linear SVM with both only achieves a test error of $0.36$, indicating these representation, while informative, do not provide easily linearly separable representations.  
KNN classifiers with other distances have errors in the range $0.22$ to $0.33$, however these are not truly comparable since we did not have the same $Q$. We set $\sigma = 0.1$ for Gaussian and Poly-kernel SVM but $\sigma = 1000$ otherwise.

\paragraph{Characters.}
The Character Trajectories Data Set from UCI Machine Learning Repository consists of handwritten characters captured using a WACOM tablet. We chose similar letters {\it p} (131 trajectories) and {\it r} (119 trajectories) (see Figure \ref{Characters}) to classify. We randomly generate 20 points as landmarks and set $\sigma = 1000$. The signed feature map $v_Q^{\sigma}$ achieves $0.01$ error rate with Random Forest and Gaussian Kernel SVM whilst $v_Q^{\mathsf{mD}}$ cannot achieve better than $0.04$ error rate. 
KNN classifier with dynamic time warping, discrete \Frechet, Hausdorff, $\dQ^{\sigma}$ and $\dQmD$ distances achieve $0$, $0$, $0.01$, $0.02$ and $0.04$ error rates, respectively.  However, $\dQ^{\sigma}$ and $\dQmD$ take about $0.25$s to build and test a KNN classifier; DTW and Hausdorff take about $35$s, and discrete \Frechet takes about $13{,}000$s.

\paragraph{Pigeons flight paths dataset.}
\cite{meade2005three} and \cite{mann2011nine,mann2014six} studied 321 recorded flight paths of pigeons repeatedly released from Horspath in the Oxford area. 
$40$ of them are shown in Figure \ref{car-bus-pigeon} (Right) in Appendix \ref{app:Q}. 
We evaluate whether we can distinguish a pigeon flight path from its reverse by randomly choosing half of the trajectories and reversing them.  
We randomly create a set of landmarks of size 20 around the curves and set $\sigma=0.05$. The unsigned minDist function $v_Q^{\mathsf{mD}}$ cannot achieve better than $0.50$ error rate; whereas our new map $v_Q^{\sigma}$ achieves $0.001$ error rate with Poly-Kernel SVM and at most $0.01$ with other classifiers.
KNN classifier with $\dQ^{\sigma}$ and dynamic time warping achieve $0.01$ and $0$ error rates, respectively, and discrete Frechet did not complete in 24 hours.  
On the other hand, non-orientation-preserving distances of $\dQmD$ and Hausdorff distance had $0.47$ and $0.42$ error rates, respectively. 

\paragraph{Directional synthetic dataset.}
Lastly, we create a synthetic dataset for which the direction information is essential.  
We generate $200$ trajectories, so half move in a clockwise direction, and half in a counter-clockwise direction. Each trajectory has $100$ waypoints.   
We define two boxes $A = [-1,1] \times [-1,1]$ and $B=[48,49] \times [0,5]$.  For each trajectory, the start and end point is chosen in $A$, and the middle waypoint (number $50$) will be in box $B$.  
For the first set of trajectories the first halfway points $2 \ldots 49$ will be in range $[-1,49] \times [0,5]$, and the second half $51 \ldots 99$ will be in $[-1,49] \times [-5,0]$.  
For the second half of the trajectories, the $y$ range of the intermediate points $2 \ldots 49$ and $51 \ldots 99$ are reversed (the first half in $[-5,0]$ and second half in $[0,5]$). We try to classify the first half from the second half. 
While $v_Q^{\mathsf{mD}}$ never achieves better than $0.48$ error rate (not much better than random), with all classifiers we achieve close to an error rate of $0$ using $v_Q^\sigma$.  
Using KNN classifiers, dynamic time warping and discrete \Frechet, can also achieve near-$0$ error rates. In this experiment $\sigma$ is set to $5$, but other values can work as well as $5$.

\newpage
\begin{table}[htbp]
\caption{Test errors with $v_Q^{\sigma}$ and $v_Q^{\mathsf{mD}}$ vectorizations.} \label{table-Beijing-Synthetic} 
{
\begin{tabular}{|c|l|cc|cc|}
\hline
& {\bf Feature Mapping} &      $v_Q^{\sigma}$ &   & $v_Q^{\mathsf{mD}}$ &  \\
\hline
&          {\bf Classifier}  &    {\bf Mean} & {\bf Std} & {\bf Mean} &  {\bf Std} 
\\ \hline 
\multirow{5}{*}{\rotatebox[origin=c]{90}{Car-Bus}} 
        & Linear SVM, $C = 1$ &       0.3611 &   0.0000  &      0.3611 &   0.0000 \\
       & Gaussian Kernel SVM , $C=1000$, $\gamma=1000$ &      0.2246 &            0.0588  &      0.3017 &     0.0640 \\
 & Poly Kernel SVM, $C=1000$, deg=auto   &   0.2453 &         0.0554  &      0.2880 &    0.0663 \\
             & Decision Tree &      0.2299 &           0.0715  &  0.2123 &   0.0687 \\
   & RandomForest with 100 estimators, max depth=7 &      \textbf{0.1567} &     0.0559  &      0.1830 &   0.0620 
\\ \hline
\multirow{5}{*}{\rotatebox[origin=c]{90} {Characters}} 
           & Linear SVM, $C=1000$ &   0.0176 &    0.0235  & 0.0401 &  0.0356  \\
  & Gaussian Kernel SVM, $C = 1$, $\gamma$ = auto &   0.0124 &  0.0151 &  0.0381 &          0.0268   \\
& Poly Kernel SVM, $C=10$, deg = auto &   0.0281 &   0.0224 &  0.0483 &  0.0345   \\
& Decision Tree &  0.0178 &   0.0284    &      0.0737 &  0.0465    \\
&  RandomForest with 100 estimators  &  \textbf{0.0100}  &  0.0138   &  0.0486 &    0.0330  
\\ \hline
\multirow{5}{*}{\rotatebox[origin=c]{90}{Pigeons}}
           & Linear SVM, $C=10000$ &          0.0025 &             0.0074  & 0.5058 &          0.0386  \\
  &Gaussian Kernel SVM, $C = 100$, $\gamma$ = auto &           0.0057 &           0.0103 &     0.5167 &          0.0388   \\
 & Poly Kernel SVM, $C=1000$, deg = auto &          \textbf{0.0010} &              0.0042 &        0.5030 &          0.0399   \\
        & Decision Tree &        0.0158 &         0.0129    &      0.5156 &  0.0435    \\
   & RandomForest with 100 estimators  &     0.0060&          0.0079   &       0.5243 &          0.0394  
\\ \hline
\multirow{5}{*}{\rotatebox[origin=c]{90}{Directional}} 
           & Linear SVM, $C=1$ &           \textbf{0.0000} &             0.0000  & 0.5116 &          0.0574  \\
  & Gaussian Kernel SVM, $C=1$, $\gamma$ = auto &           \textbf{0.0000} &           0.0000 &     0.5116 &          0.0527   \\
 & Poly Kernel SVM, $C=1$, deg = auto &          \textbf{0.0000} &              0.0000 &        0.4864 &          0.0525   \\
        & Decision Tree &        0.0036 &         0.0088    &      0.5136 &  0.0608    \\
   &  RandomForest with 100 estimators  &     \textbf{0.0000} &          0.0000   &       0.5032 &          0.0584   \\
\hline 
\end{tabular}
}
\end{table}

\newpage

\bibliographystyle{plain}
\bibliography{references}

\newpage
\appendix

\section{Technical Details on Defining the Normal and Computing $v_q(\gamma)$} \label{app:alg} 

We need to assign a normal vector to each point of a curve $\gamma$. Let $\gamma\in \Gamma'$ and $q \in \mathbb{R}^2$. Assume $p= \argmin_{p'\in \gamma} \|q-p'\|$. If $p \in \reg(\gamma)$, as we mentioned earlier, according to the right hand rule, we can assign a unit normal vector to $\gamma$ at $p$ which is compatible with the direction of $\gamma$. We also assign a fixed normal vector at endpoints of a non-closed curve as we can use the normal vector of tangent line at endpoints that is compatible with the direction of curve. Now it remains to define a normal vector at critical points. We must be careful about doing this as any vector can be considered as a normal vector at critical points. Our aim is to define a unique normal vector at critical points with respect to a landmark point. Assume $p$ is not an endpoint of $\gamma$. Denote by $N(p)$ the closure of the set of all unit vectors $u$ such that $u$ is perpendicular to $\gamma$ at $p$ and is compatible with the direction of $\gamma$ by the right hand rule (for instance, in Figure \ref{Fig15}(a), $N(p)$ is the set of all unit vectors between $n$ and $n'$). 
Notice that for regular points on the curve $N(p)$ is a singleton and indeed, it does not depend on $q$ but only on the direction of the curve. Then we define 
\[
n_{p}(q)= \argmax \big\{|\langle u, q-p\rangle|: u \in N(p)\big\}.
\]
It can be readily seen that $n_p(q) = \sign(q,p,\gamma) \frac{q-p}{\|q-p\|}$, where $\sign(q,p,\gamma)$ can be obtained via Algorithm \ref{alg1}.
At endpoints we fix a normal vector, the one that is perpendicular to the tangent line. Notice, as the range of $\gamma$ is a compact subset of $\mathbb{R}^2$ and the norm function is continuous, the point $p$ exists. Continuity of the inner product and compactness of $N(p)$ guarantee the existence of $n_{p}(q)$ as well. A different algorithmic approach for finding $n_p(q)$, when $p$ is a critical point, is given below.

\subsection{Computing $n_p(q)$ and the Sign}

Assume that a trajectory $\gamma$ is given by the sequence of its critical points (including endpoints) $\{c_i\}_{i=0}^n$ and $\|c_i-c_{i+1}\|>0$ for each $0\leq i \leq n-1$. The following algorithm determines the sign of a landmark point $q$ with respect to $\gamma$ when $p = \argmin\{\|q-x\|: x \in \gamma\}$ is a critical point of $\gamma$. 

\begin{algorithm}
\caption{Find $\sign(q, p, \gamma)$}
\label{alg1}
\begin{algorithmic}
\STATE {\textbf{Input:} A landmark $q\in \mathbb{R}^2$, a trajectory $\gamma = \{c_i\}_{i=0}^n$ and $p = c_i$ for some $0\leq i \leq n$.}
\STATE Find $w_i$ (unit normal vector to the directed segment $\overrightarrow{c_i c_{i+1}}$).
\STATE Find $\alpha$ (the angle between $w_{i-1}$ and $w_i$)
\STATE Find $\theta$ (the angle between $w_{i-1}$ and $q-p$)
\STATE Put $t = \frac{1}{2}(1 - \cos(\frac{\pi \theta}{\alpha}))$
\STATE Let $n_t$ be the normalized convex combination of $w_{i-1}$ and $w_i$ by $t$\;
\RETURN $\sign(\langle n_t, q-p \rangle)$.
\end{algorithmic}
\end{algorithm}
Because each step in Algorithm \ref{alg1} takes constant time, it is clear that the running time is $O(1)$.  
Now, in Algorithm \ref{alg2}, for a landmark point $q\in \mathbb{R}^2$ and a trajectory $\gamma$ we provide steps to compute the sketch vector $v_q(\gamma)$.

\begin{algorithm}
\caption{Compute $v_q(\gamma)$}
\label{alg2}
\begin{algorithmic}
\STATE {A landmark $q \in \mathbb{R}^2$ and a trajectory $\gamma = \{c_i\}_{i=0}^n$}
\FOR {$i = 0, \ldots, n-1$}
\STATE Find $d_i$, the distance of $c_i$ from segment $S_i =  \overrightarrow{c_i c_{i+1}}$
\ENDFOR
\STATE Set $j = \argmin\{d_i: 0 \leq i \leq n-1\}$
\STATE Set $p = \argmin\{\|q-x\|: x \in S_j\}$
\STATE Using Algorithm \ref{alg1} compute $v_q(\gamma)$
\RETURN $v_q(\gamma)$
\end{algorithmic}
\end{algorithm}
It can be readily seen that Algorithm \ref{alg2} can be run in linear time  in terms of the size of $\gamma$. 

\subsection{Defining the Normal: A Computational Approach}

If we look at self-crossing curves, for instance, we will notice that the landmark point $q$ will opt for the crossing point $p$ only if tangent lines of $\gamma$ at $p$ (where $q$ lies in it) make an angle $\beta \geq \pi$ (see Figure \ref{Fig11}). Therefore, there is no need to define a normal vector for all crossing points and without loss of generality we may assume that $\beta \geq \pi$. 

\begin{figure}[h] 
\centering
\includegraphics{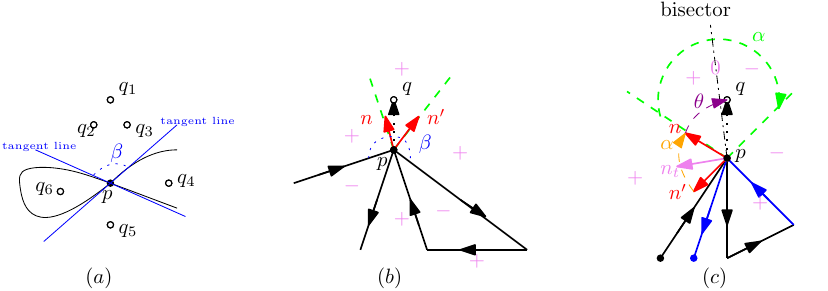}
\caption{Choice of self-crossing point}
\label{Fig11}
\end{figure} 

Let $\gamma$ be a curve with $\beta \geq \pi$ at a crossing point $p$ (Figure \ref{Fig11}(b),(c)). For $t\in [0,1]$ we consider $n_t  = \frac{(1-t) n + t n'}{\|(1-t) n + t n'\|}$, where $n$ and $n'$ are normal vectors to the curve at a crossing point $p$. It is necessary to agree that $n_{1/2} = 0$ if $n'= -n$ which is possible when $\beta = \pi$. 
Now the question is how to choose the parameter $t$? Let $\alpha$ be the angle between $n$ and $n'$ and $\theta$ be the angle between $n$ and $q-p$ (as shown in Figure \ref{Fig11}(b),(c)), i.e. 
$\alpha = \arccos(\langle n, n'\rangle)$ and $\theta = \arccos(\langle n, \frac{q-p}{\|q-p\|}\rangle).$
Then $0 < \alpha \leq \pi$ and $0 \leq \theta \leq \alpha$ and thus $0 \leq \frac{\pi}{\alpha} \theta \leq \pi$. Now we can set $n_t = (1-\cos(\frac{\pi}{\alpha} \theta))/2$. 
Therefore, the following hold:
\begin{enumerate}
\item If $\theta = 0$, then $t=0$, $q$ is on the left dashed green line in Figure \ref{Fig11}(c) and $n_t = n = \frac{q-p}{\|q-p\|}$. 
\item Moving towards the bisector, $n_t$ rotates towards $n'$ and so $\theta$ increases (but still $\theta < \pi/2$). Hence $\langle n_t, q-p \rangle$ is positive and is decreasing as a function of $\theta$. 
\item When $\theta = \frac{\alpha}{2}$, $t=\frac{1}{2}$ and $q$ is on the bisector of $\theta$ and $\langle n_t, q-p \rangle = 0$ (Figure \ref{Fig11}(c)). 
\item Moving from bisector towards the other side of the green dashed angle, $\theta$ increases and $\theta > \pi/2$. Thus $\langle n_t, q-p \rangle$ is negative and decreases. 
\item If $\theta = \alpha$, then $t=1$, $q$ is on the right dashed green line in Figure \ref{Fig11}(c) and $n_t = n' = - \frac{q-p}{\|q-p\|}$. 
\end{enumerate}
However, we will only need to compute the inner product of $n_t$ and $q-p$, which can easily be obtained by $\langle n_t, q-p \rangle = \|q-p\| \cos(\frac{\pi}{\alpha} \theta)$.

The way we defined $n_t$ is a general rule for any crossing and critical point. Now we are going to clarify obtaining $n_t$ in different situations. 

\begin{enumerate}
\item Let $p$ be a critical point which is not a crossing point of $\gamma$. Then as above a landmark point $q$ will choose $p$ as an $\argmin$ point only if tangent lines at $p$  constitute an angle $\beta \geq \pi$ and $q$ is inside of that area (Figure \ref{Fig15}(a)). In this case we can easily see that $n_t = \frac{q-p}{\|q-p\|}$ for any $t$. 
\item In self-crossing case of Figure \ref{Fig11}(b), again we can observe that $n_t = \frac{q-p}{\|q-p\|}$ for any $t$.
\item If $p$ is an end point, we consider $n$ as the normal vector at $p$ to the tangent line at $p$ to $\gamma$ and we set $n' = -n$. Then for $t \in [0,\frac{1}{2})$, $n_t = n$, $n_{1/2} = 0$ and for $t \in (\frac{1}{2},1]$, $n_t = -n$ (see Figure \ref{Fig15}(b)). 
\end{enumerate}

As we saw above, $n_t$ depends upon the landmark point $q$, that is, a critical point $p$ can be an $\argmin$ point for many landmark points $q$. Therefore, we use the notation $n_p(q)$ instead of $n_t$. For $p \in \reg(\gamma)$ we set $n_p(q) = n_p$ for any $q$ such that $p = \argmin_{p' \in \gamma} \|q - p'\|$. 

\begin{figure}[h] 
\centering
\includegraphics[width=0.7 \textwidth]{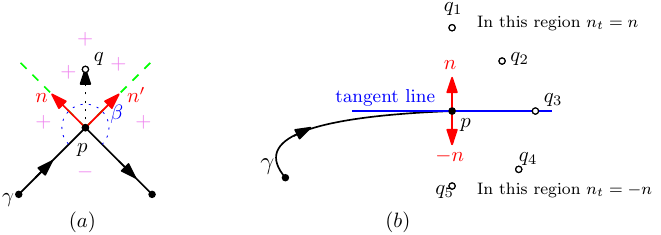}
\caption{Matching of $n_t$ with $n_p(q)$}
\label{Fig15}
\end{figure} 

\subsection{Trajectory plots and landmarks $Q$ used in Section \ref{S4}} \label{app:Q}

The landmarks $Q$ used in Car-Bus classification in Section \ref{S4} is \\
$Q = \{[-11.51839252, -36.00167761],
       [-13.82387038, -41.06091863],
       [ -5.52315713, -37.39424807], $ 
       
$\hspace{0.4cm}[-11.307905  , -37.68847429],
       [-13.34456129, -36.62837387],
       [-11.25127849, -38.96427889],$
       
$\hspace{0.4cm}       [ -9.70920127, -37.70032869],
       [-10.84289516, -36.12955072],
       [-13.29665413, -37.43363195],$
       
$\hspace{0.4cm}       [ -9.06848519, -37.87890405],
       [-11.1337595 , -36.73167919],
       [-11.15004646, -37.94099833],$
       
$\hspace{0.4cm}       [-11.76909718, -33.96734203],
       [-11.29377414, -40.15240307],
       [-11.38400841, -37.02645598],$

$\hspace{0.4cm}       [-10.64689404, -37.77147543],
       [-11.20407876, -36.41339879],
       [ -9.16302112, -37.81380282],$

$\hspace{0.4cm}       [-12.87662174, -38.41711695],
       [-12.27654292, -38.37274434]\}$. \\

The landmarks $Q$ used in characters dataset classification in Section \ref{S4} is \\
$Q = \{ [-32.4902653,  46.0044218],
       [-16.8353212,  4.43267019],
       [-110.587361,  25.9348076],$
       
$\hspace{0.4cm} [ 29.2457830, -39.9858732],
       [ 11.7810639, -77.9407977],
       [ 35.3840384, -12.5123595],$
       
$\hspace{0.4cm} [ 12.6823738, -19.9258348],
       [ 54.2110442,  31.0153988],
       [ 1.52758301,  3.62917151],$
       
$\hspace{0.4cm} [-46.2751522, -2.46585474],
       [ 32.5632345, -13.4574254],
       [ 3.96964624,  6.12114622],$
       
$\hspace{0.4cm} [-47.1061489,  17.8451989],
       [ 33.0163320,  4.06856467],
       [-39.1222051, -41.9573723],$
       
$\hspace{0.4cm} [-59.0871676,  62.9764287],
       [ 32.1785601,  72.9091697],
       [-48.3503786,  46.1095240],$
       
$\hspace{0.4cm} [ 25.0696612, -57.4618787],
       [ 8.60611470,  36.9777691] \}.$ \\

The landmarks $Q$ used in pigeons dataset classification in Section \ref{S4} is \\
$Q = \{[51.84008709, -1.20588467],
       [51.72283931, -1.26982576],
       [51.80884313, -1.27107579],$
       
       $\hspace{0.4cm} [51.70708245, -1.18863574],
       [51.76219485, -1.19758355],
       [51.74535629, -1.24242755],$
       
       $\hspace{0.4cm} [51.75137002, -1.26676556],
       [51.76411002, -1.3096359 ],
       [51.69815133, -1.26414871],$
       
       $\hspace{0.4cm} [51.75161713, -1.24502458],
       [51.75894275, -1.29741779],
       [51.71222909, -1.10956269],$
       
       $\hspace{0.4cm} [51.84580983, -1.24300095],
       [51.72781533, -1.1806641 ],
       [51.75728962, -1.27375996],$
       
       $\hspace{0.4cm} [51.6389533 , -1.24789817],
       [51.82411354, -1.13822127],
       [51.84027399, -1.2980129 ],$
       
       $\hspace{0.4cm} [51.75530875, -1.28604281],
       [51.84218647, -1.2052331] \}$. 
       
\newpage
       
The landmarks $Q$ used in synthetic dataset classification in Section \ref{S4} is \\
$Q = \{[31.25226115, -7.06554558],
       [15.5893362 ,  0.53644793],
       [46.0298242 ,  3.99239995],$
       
$\hspace{0.4cm} [42.01539614,  6.32997671],
       [47.38055911,  1.55938833],
       [20.77769647, -7.38738266],$

$\hspace{0.4cm}[ 2.49272512,  3.29996588],
       [22.88248179,  4.20715546],
       [11.45275598, -1.32409368],$

$\hspace{0.4cm} [44.07861775,  4.51315021],
       [45.82336144,  4.15814764],
       [22.84178316, -3.18598697],$
       
$\hspace{0.4cm} [11.455545  ,  1.70912469],
       [-0.63639767, -5.20555862],
       [14.28345207, -0.26886728],$

$\hspace{0.4cm}  [50.86332419,  4.82330847],
       [39.72143511,  6.7231165 ],
       [10.75453184, -3.27418179],$
       
$\hspace{0.4cm} [16.01803815,  1.11604296],
       [43.00847275, -4.83131665]\}$. \\ 
       
 \begin{figure}[h] 
\centering
\includegraphics[width=0.48 \textwidth]{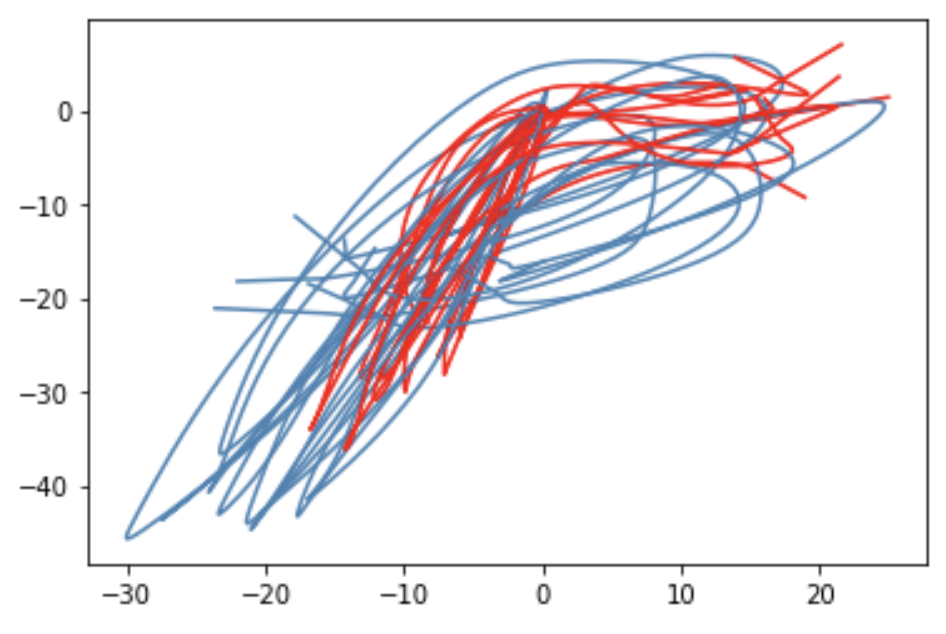}
\caption{10 trajectories (out of 131) for the letter {\it p} (blue) and 10 (out of 119) for the letter {\it r} (red).}
\label{Characters}
\end{figure} 

\begin{figure}[h] 
\centering
\includegraphics[width=0.31 \textwidth]{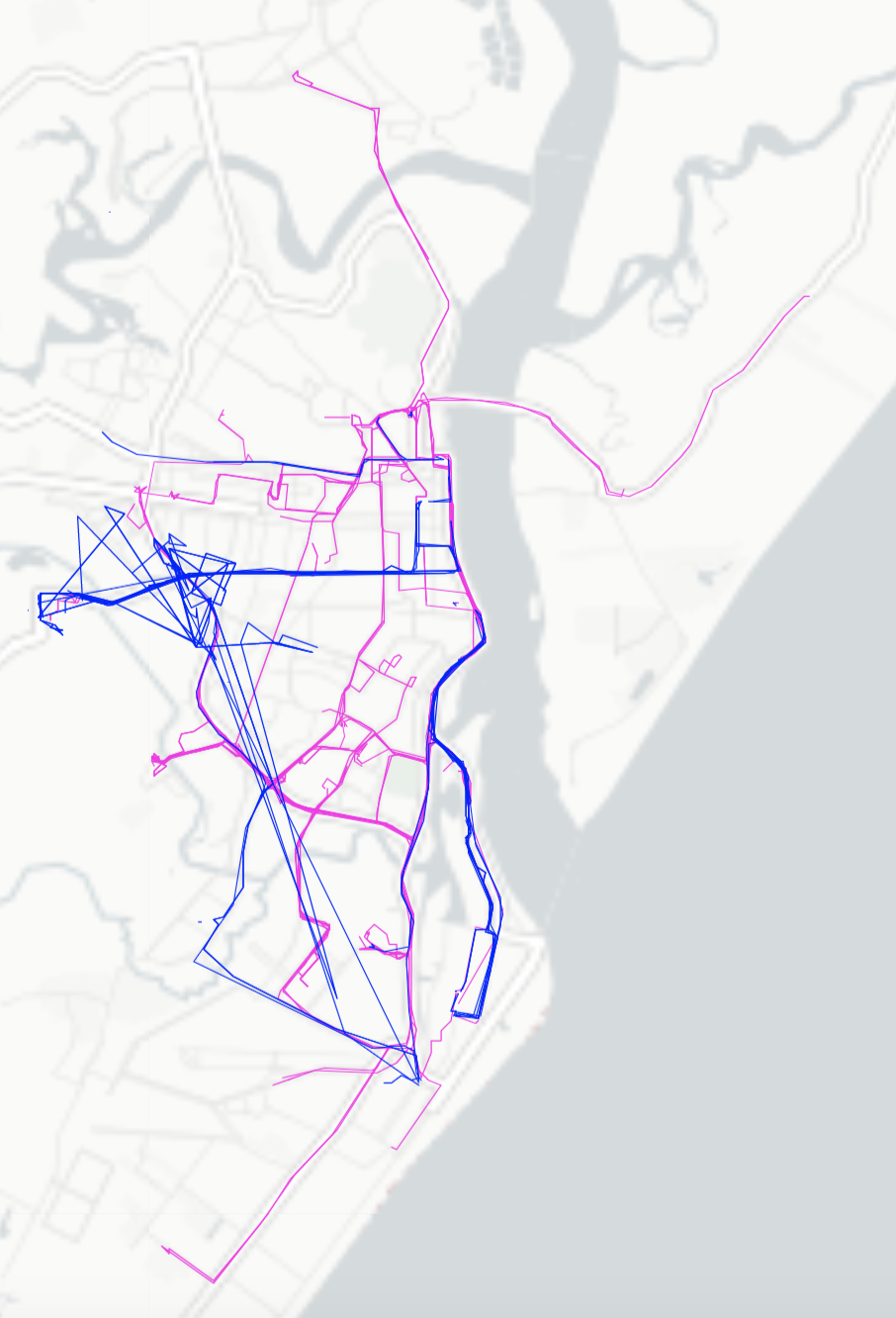} \hspace{0.5cm}
\includegraphics[width=0.405 \textwidth]{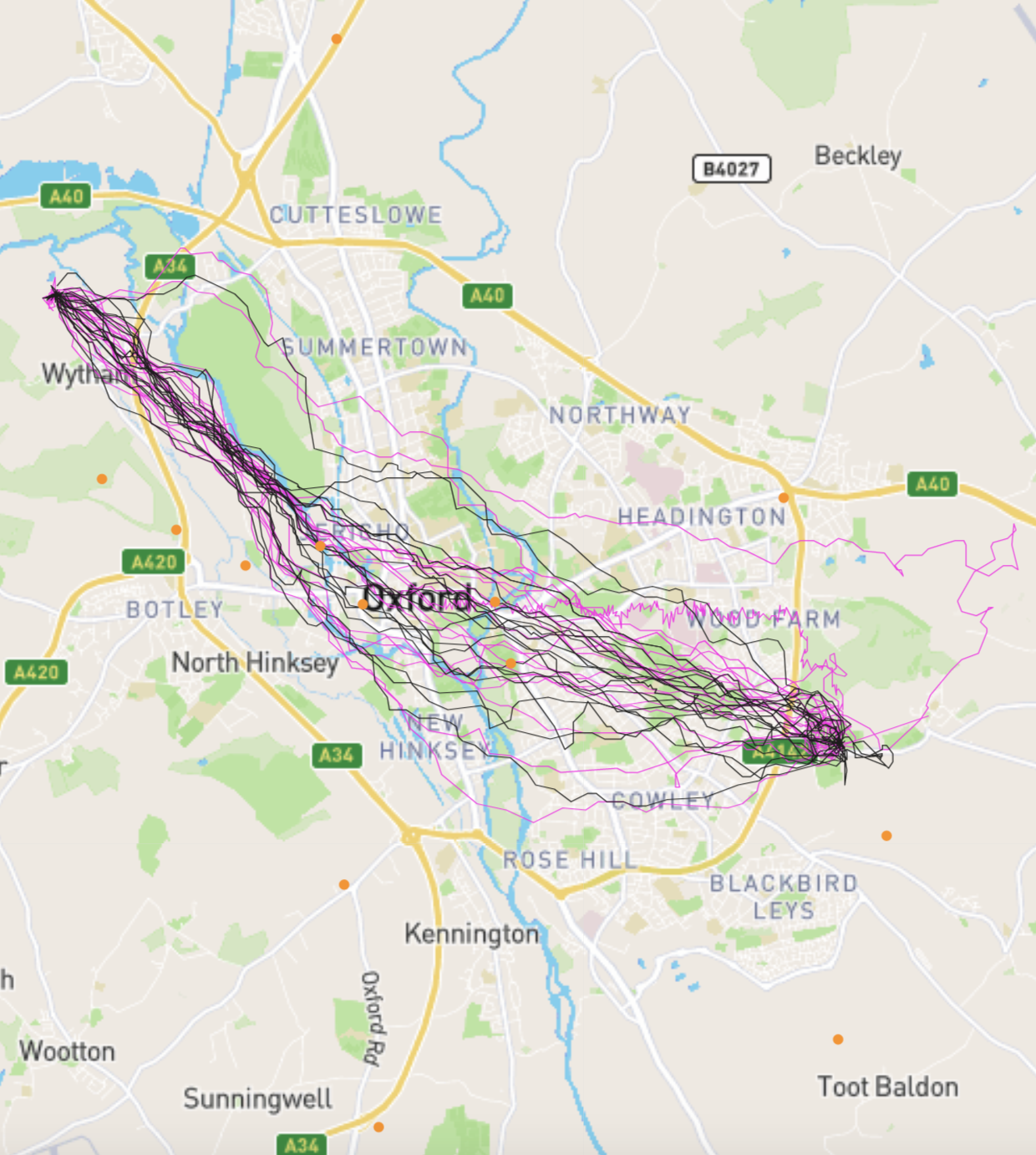}
\caption{Left: Bus (blue) and car (pink) trajectories from the GPS Trajectories Data Set \cite{GTDS2016} recorded in Aracuja, Brazil. Since landmarks are a bit far from trajectories, we have not shown them on the map. Right: 20 (out of 160) pigeons flight paths in one direction (pink) and 20 (out of 161) in the opposite direction (black) chosen from Horspath dataset recorded in Oxford area, with chosen landmarks (orange bullets).}
\label{car-bus-pigeon}
\end{figure}

The codes in Python and the car-bus dataset are stored in the following GitHub repository: \\
\url{https://github.com/aghababa/Orientation_Preserving_Distance_MSML2021}
       
\end{document}